\documentclass{article}

\usepackage{color, xcolor, colortbl}
\usepackage{graphicx}
\usepackage{geometry}
\usepackage{amsmath,amssymb,amsfonts,amsthm}
\usepackage{mathtools}
\usepackage{algorithm}
\usepackage{algorithmic}
\usepackage{dcolumn}
\usepackage{epstopdf}
\usepackage{bm}
\usepackage[caption=false]{subfig}
\usepackage{appendix}
\usepackage{multirow}
\usepackage{braket}
\usepackage[english]{babel}
\usepackage{cite}
\usepackage{aliascnt}
\usepackage{hyperref}
\usepackage[capitalize]{cleveref}
\usepackage{xpatch}
\usepackage{tikz}
\usepackage{adjustbox}
\usepackage{xspace}

\newcommand{\ugamma}{u_{\gamma}}
\newcommand{\Fzerogamma}{F_{0,\gamma}}
\newcommand{\Fonegamma}{F_{1,\gamma}}
\newcommand{\Ftwogamma}{F_{2,\gamma}}

\renewcommand{\Re}{\operatorname{Re}}

\newcommand{\Tr}{\operatorname{Tr}}

\let\originalleft\left
\let\originalright\right
\renewcommand{\left}{\mathopen{}\mathclose\bgroup\originalleft}
\renewcommand{\right}{\aftergroup\egroup\originalright}

\newtheorem{theorem}{Theorem}[section]

\newaliascnt{lemma}{theorem}
\newtheorem{lemma}[lemma]{Lemma}
\aliascntresetthe{lemma}

\crefname{theorem}{Theorem}{Theorems}
\Crefname{theorem}{Theorem}{Theorems}
\crefname{lemma}{Lemma}{Lemmas}
\Crefname{lemma}{Lemma}{Lemmas}

\providecommand{\definitionname}{Definition}
\providecommand{\assumptionname}{Assumption}
\providecommand{\corollaryname}{Corollary}
\providecommand{\lemmaname}{Lemma}
\providecommand{\propositionname}{Proposition}
\providecommand{\remarkname}{Remark}
\providecommand{\theoremname}{Theorem}

\title{Fast-forwarding quantum algorithms for weakly nonlinear dissipative
differential equations and beyond}

\author{
Yixiang Li\thanks{School of Mathematical Sciences, Peking University,
Beijing, China.}
\and
Dong An\thanks{Beijing International Center for Mathematical Research,
Peking University, Beijing, China.
\href{mailto:dongan@pku.edu.cn}{dongan@pku.edu.cn}.}
}

\date{}

\begin{document}

\maketitle

\begin{abstract}
    We study a fast-forwarded quantum algorithm for solving weakly nonlinear dissipative ordinary differential equations. 
    Our approach is a combination of the Carleman embedding technique and the linear combination of Hamiltonian simulation algorithm for linearized systems with fast-forwarded scaling. 
    The complexity of our algorithm does not explicitly depend on the evolution time $T$, thus greatly improving the previous state-of-the-art $\widetilde{\mathcal{O}}(\sqrt{T})$ to $\mathcal{O}(1)$, and any remaining time dependence enters through the output norm and forcing parameters. 
    We rigorously analyze the performance of this approach by convergence guarantees of the Carleman embedding for time-dependent coefficient matrices and detailed complexity estimates, and improve the realization of the Carleman-embedding-based algorithms by simplifying the post-selection step. 
    In addition, we perform a numerical study on differential equations beyond the weakly nonlinear case, and identify possibility of achieving fast-forwarding scaling for systems with stronger nonlinearity or linear non-resonant effect. 
\end{abstract}

\section{Introduction}

The mathematical modeling of complex dynamical phenomena across the physical, biological, and chemical sciences is governed by systems of nonlinear differential equations. 
Solving large-scale systems of nonlinear differential equations remains a significant computational challenge.
As the system dimension or required spatial resolution increases, classical numerical integrators suffer from polynomial or exponential complexity scaling, driving the need for alternative computational paradigms. 

Quantum computing provides a promising computational framework to potentially resolve the issue of high dimensionality, and recent years have witnessed rapid developments in the design of efficient quantum algorithms for nonlinear differential equations~\cite{LiuKoldenKroviEtAl2021,CostaSchleichMoralesBerry2025,WuWangLi2025,JenningsKorzekwaLostaglioEtAl2025,WangJiaVeerapaneniDing2026,JinLiu2024,JinLiu2026HJB,Joseph2020,katzMuraleedharanAlase2025,JenningsKorzekwaLostaglioWang2026,JinLiuMedvidovaYuan2026}. 
Typically, these algorithms first linearize the differential equations, and then apply the quantum linear differential equation algorithms~\cite{Berry2014,BerryChildsOstranderEtAl2017,ChildsLiu2020,Krovi2022,BerryCosta2022,FangLinTong2022,JinLiuYu2024,AnLiuLin2023,AnChildsLin2023,ShangGuoAnZhao2025,LowSomma2025,Li2026}. 
One of the most commonly used linearization techniques is the Carleman embedding. 
Remarkably, quantum algorithms based on Carleman embedding can solve, with provable convergence and efficiency guarantee, certain types of weakly nonlinear equations~\cite{LiuKoldenKroviEtAl2021,CostaSchleichMoralesBerry2025,WuWangLi2025,JenningsKorzekwaLostaglioEtAl2025,WangJiaVeerapaneniDing2026}. 

Despite these profound advances, most of the quantum nonlinear differential equation algorithms scale linearly (or even worse) in the evolution time $T$. 
This means increasing computational complexity when simulating long-time dynamics and restricts the feasibility of long-term asymptotic predictions. 
Furthermore, generic quantum differential equation algorithms remain fundamentally constrained by the ``no-fast-forwarding theorem''~\cite{BerryChildsCleveEtAl2015}, which dictates that worst-case computational complexity must scale at least linearly with $T$. 

To further accelerate the long-time simulation, one possibility is to design sublinear-time-scaling \emph{fast-forwarded} quantum algorithms applicable to only a subset of differential equations of practical interest, thus bypassing the linear-time constraint imposed in the worst case. 
Fast-forwarding has been made possible and well studied in the context of Hamiltonian simulation~\cite{AtiaAharonov2017,GuSommaSahinoglu2021} and solving linear differential equations~\cite{JenningsLostaglioLowrieEtAl2024,YangOnwuntaAn2025,AnLiuWangEtAl2022,HuJin2026}. 
However, there have been fewer works~\cite{JenningsKorzekwaLostaglioEtAl2025,WangJiaVeerapaneniDing2026} on fast-forwarding quantum nonlinear differential equation algorithms, and it remains unclear to what extent such algorithms can be fast-forwarded. 

In this work, we focus on the design and analysis of quantum nonlinear differential equation algorithms with fast-forwarded time complexity. 
We consider the prototype system of quadratic nonlinear ODEs 
\begin{equation}\label{eqn:nonlinear_ODE}
    \frac{du(t)}{dt} = F_2(t) u(t)^{\otimes 2} + F_1(t) u(t) + F_0(t). 
\end{equation}
Here $u = [u_0;u_1;\cdots;u_{n-1}] \in \mathbb{C}^{n}$, $u^{\otimes 2} = u\otimes u \in \mathbb{C}^{n^2}$, and each $F_j(t) \in \mathbb{C}^{n\times n^j}$ is a time-dependent coefficient matrix or vector. 
Quantum algorithms for~\cref{eqn:nonlinear_ODE} aim to prepare a quantum state encoding a normalized $u(T)$ in its amplitudes, which can be measured and classically postprocessed to estimate, for suitably structured observables, useful quantities such as component probabilities, quadratic observables $\bra{u(T)}M\ket{u(T)}$, fidelities with reference solutions, and correlation functions and variances~\cite{BrassardHoyerMoscaEtAl2002,HuangKuengPreskill2020,ElbenFlammiaHuangEtAl2023}.

\subsection{Contribution}
The core contribution of our work is a fast-forwarded quantum algorithm for weakly nonlinear dissipative ODEs. 
Specifically, under the assumption that the linear part of~\cref{eqn:nonlinear_ODE} is dissipative and the non-linear part is relatively weaker than the linear dissipation, we show that~\cref{eqn:nonlinear_ODE} can be solved with complexity explicitly independent of the evolution time $T$. 

Our algorithm follows the standard procedure that we first linearize~\cref{eqn:nonlinear_ODE} by the Carleman embedding~\cite{LiuKoldenKroviEtAl2021,JenningsKorzekwaLostaglioEtAl2025} and solve the linearized system using the linear combination of Hamiltonian simulation (LCHS) algorithm~\cite{YangOnwuntaAn2025}. 
We establish a rigorous convergence analysis of the Carleman embedding in the case of time-dependent coefficient matrices and show that the linearized system is strictly dissipative, thereby enabling the fast-forwarded complexity of LCHS shown in~\cite{YangOnwuntaAn2025}. 
These together result in the elimination of the explicit time dependence in solving the nonlinear ODEs. 

We also improve the post-selection step in the design of the algorithm. 
Carleman embedding is achieved by considering the dynamics of an enlarged solution involving $u(t)$ as well as its higher-order terms $u(t)^{\otimes j}$. 
Previous works usually regard higher-order terms as useless parts and post select the enlarged solution state onto its first component to get $u(t)$. 
In our work, we notice that higher-order terms also contain the information of the solution $u(t)$, and the corresponding Carleman embedding errors remain controlled. 
Therefore, our approach directly discards the ancilla register used for enlarging the space without post selection, and we show that the resulting state still remain a valid numerical solution of the original nonlinear ODE. 

Besides dissipative linear part, Carleman embedding has been proven to converge for other weakly nonlinear systems, such as those with partial conservation~\cite{JenningsKorzekwaLostaglioEtAl2025} or non-resonant linear part~\cite{WuWangLi2025,JenningsKorzekwaLostaglioEtAl2025}. 
In this work, we also numerically explore the possibility of fast-forwarding by examining the Carleman embedding convergence and the eigenvalues of the linearized coefficient matrix. 
Our numerical results indicate that Carleman embedding can still converge with stronger nonlinearity than theoretically predicted, which has already been observed in existing works~\cite{LiuKoldenKroviEtAl2021}. 
Additionally, for dissipative nonlinear ODEs, all of the tested linearized ODEs are still dissipative, which implies that our fast-forwarded algorithm might apply to dissipative ODEs with stronger nonlinearity in practice. 
On the other hand, for conservative and non-resonant systems, although Carleman embedding may also converge in scenarios beyond theoretical prediction, all of our tested linearized ODEs fail to be dissipative, so our fast-forwarded algorithm does not work. 
Nevertheless, we find that the linearized ODEs can be asymptotically dissipative when the coefficient matrix $F_1$ is time-independent and all its eigenvalues have negative real parts. 
In this case, we can alternatively apply the contour-integral-based quantum algorithm for linear ODEs~\cite{TakahiraOhashiSogabeEtAl2020,TakahiraOhashiSogabeEtAl2021,JiangAn2026} to solve the nonlinear ODEs with complexity explicitly independent of the evolution time $T$ as well.

\subsection{Related works}
Several existing works~\cite{ForetsPouly2017,LiuKoldenKroviEtAl2021,AminiZhengSunEtAl2022,WuWangLi2025,JenningsKorzekwaLostaglioEtAl2025} have carefully studied and established the convergence of Carleman embedding with time-independent coefficients. 
Our convergence analysis for the time-dependent case follows the theory established in~\cite{JenningsKorzekwaLostaglioEtAl2025} with a more explicit parameter dependence. 
Additionally, our studies on the conservative and non-resonant systems are motivated from~\cite{JenningsKorzekwaLostaglioEtAl2025} as well as~\cite{WuWangLi2025}, which first establishes the Carleman embedding convergence for the weakly nonlinear non-resonant systems. 

For weakly nonlinear dissipative ODEs, most of the existing quantum algorithms scale linearly or even worse in $T$, and there are two works on fast-forwarding scalings. 
Reference~\cite{JenningsKorzekwaLostaglioEtAl2025} shows that the history state (i.e., a state encoding all intermediate steps along the evolution) can be obtained with cost logarithmically in $T$, but extracting the final solution from the history state can incur an $\mathcal{O}(\sqrt{T})$ scaling due to the post-selection success probability. 
Recently, \cite{WangJiaVeerapaneniDing2026} develops a pivot-shifted Carleman embedding technique and shows that the final solution of the ODE under milder stability condition can be obtained with cost $\mathcal{O}(\sqrt{T})$. 
As a comparison, the complexity of our fast-forwarded quantum algorithm does not explicitly depend on $T$, thereby further eliminating the $\widetilde{\mathcal{O}}(\sqrt{T})$ time dependence for final state preparation. 

As discussed earlier, another improvement of our work is the simplification of the post-selection step. 
Prior Carleman-based quantum algorithms commonly project the enlarged solution $[u;u^{\otimes 2};u^{\otimes 3}; \cdots]$ onto its degree-one sector, which succeeds with probability roughly $(1-\|u\|^2)$ and introduces an additional $\frac{1}{\sqrt{1-\|u\|^2}}$ multiplicative factor with amplitude amplification. 
Our trick directly discards the ancilla rather than post selection, and thus avoids this additional $\frac{1}{\sqrt{1-\|u\|^2}}$ factor in the complexity.

\subsection{Organization}
The rest of this paper is organized as follows. 
Section~\ref{sec:setup} formally defines the setup and stability conditions for the quadratic nonlinear ODEs. 
In \cref{sec:linearization}, we detail the Carleman embedding procedure with time-dependent coefficients and establish rigorous bounds on the truncation error, and then present the design of the fast-forwarded quantum algorithm with a comprehensive theoretical complexity analysis in~\cref{sec:algorithm}. 
We show our numerical results and briefly introduce an alternative fast-forwarding approach in~\cref{sec:numerics}, and conclude with future works in~\cref{sec:concludion}.

\section{Setup}\label{sec:setup}

Consider a system of quadratic nonlinear differential equations in~\cref{eqn:nonlinear_ODE}. 
In this work, we mainly consider the weakly nonlinear dissipative differential equations, which satisfy the following conditions:

\begin{enumerate}
    \item (Dissipation) the coefficient matrix $F_1(t)$ of the linear term has uniformly negative logarithmic norm, i.e., 
    \begin{equation*}
        \mu(F_1) : = \sup_{t\in [0,T]}\max_{\ket{x} = 1}\Re \braket{x|F_1(t)|x} < 0. 
    \end{equation*}
    \item (Weak nonlinearity)
    \begin{equation*}
        R : = \frac{1}{-\mu(F_1)}\left(\|F_2\|\|u(0)\| + \frac{\|F_0\|}{\|u(0)\|}\right) < 1, 
    \end{equation*}
    where $\|F_2\| = \sup_{t \in [0,T]} \|F_2(t)\|$ and  $\|F_0\| = \sup_{t\in [0,T]}\|F_0(t)\| $. \footnote{More generally, in this work the notation $\|\cdot\|$ denotes the matrix/vector $2$ norm when the argument is a fixed matrix/vector, and the supremum of the norms over the time interval $[0,T]$ when the argument is a matrix-/vector-valued function. }
\end{enumerate}

    Following~\cite{JenningsKorzekwaLostaglioEtAl2025}, we can extend our setting to Lyapunov stable ODEs, where the dissipation condition on the logarithmic norm of $F_1(t)$ is replaced by the condition that there exists a time-independent positive-definite Lyapunov matrix \(P>0\) such that $F_1(t)^\dagger P + P F_1(t) \le - 2c P$ for some constant \(c>0\) and all $t \in [0,T]$. 
    In this case, we can use the linear change of variables \(\widetilde{u}(t) = Qu(t)\) with \(Q = P^{1/2}\), and the new variable $\widetilde{u}(t)$ satisfies 
    \begin{equation}\label{eqn:nonlinear_ODE_transformed}
        \frac{d \widetilde{u}(t)}{dt} = \widetilde{F}_0(t) + \widetilde{F}_1(t) \widetilde{u}(t) + \widetilde{F}_2(t) \widetilde{u}(t)^{\otimes 2}, 
    \end{equation}
    where $\widetilde{F}_0(t) = Q F_0(t)$, $\widetilde{F}_1(t) = Q F_1(t) Q^{-1}$, and $\widetilde{F}_2(t) = Q F_2(t) (Q^{-1}\otimes Q^{-1})$. 
    Therefore, all the theoretical results for dissipative equations can be extended to~\cref{eqn:nonlinear_ODE_transformed} as long as $\widetilde{F}_1(t)$ has uniformly negative logarithmic norm and the $R$ parameter associated with~\cref{eqn:nonlinear_ODE_transformed} is bounded by $1$. 
    Our fast-forwarded quantum algorithm to be discussed later also applies if a block-encoding of $Q$ is available.

\section{Carleman embedding with time-dependent coefficients}\label{sec:linearization}

In this section, we will first show how to linearize the nonlinear differential equations with time-dependent coefficients by Carleman embedding, and then prove the convergence result of this procedure. 
The approach and theory are essentially the same as the time-independent coefficient case established in~\cite{LiuKoldenKroviEtAl2021,JenningsKorzekwaLostaglioEtAl2025} but with a more explicit parameter dependence. 
For completeness, here we present a self-contained discussion for the time-dependent coefficient case. 

\subsection{Linearization procedure}\label{sec:Carleman_linearization_procedure}

The first step of the Carleman embedding is to rescale the original ODE. 
Specifically, we introduce a rescaling factor $\gamma$ to rewrite~\cref{eqn:nonlinear_ODE} as 
    \begin{equation*}
        \gamma \frac{du(t)}{dt} = \gamma F_0(t) + F_1(t) \gamma u(t)+ \frac{F_2(t)}{\gamma}  (\gamma u(t)  )^{\otimes 2}.
    \end{equation*}
    Thus,~\cref{eqn:nonlinear_ODE} is equivalent to the ODE system 
    \begin{align*}
    \dot{u}_{\gamma}(t) &= {F}_{0,\gamma}(t)+ {F}_{1,\gamma}(t)u_{\gamma}(t)+ F_{2,\gamma}(t) u_{\gamma}(t)^{\otimes 2},
    \end{align*}
    where
    \begin{align*}
        u_{\gamma}(t) = \gamma u(t), \quad F_{0,\gamma} (t) = \gamma F_0(t), \quad
        F_{1,\gamma}(t) = F_1(t), \quad
        F_{2,\gamma}(t) =  \frac{1}{\gamma}F_2(t). 
    \end{align*}
    The purpose of such rescaling is to impose additional technical conditions for the convergence of Carleman embedding, which will be clear later. 
    We remark that the rescaling technique can also improve the efficiency of quantum nonlinear differential equation algorithms (see~\cite{CostaSchleichMoralesBerry2025} for a careful study). 

Next, consider the infinite-dimensional vector $[\ugamma;\ugamma^{\otimes 2};\ugamma^{\otimes 3}; \cdots]$. 
It is straightforward to verify that this vector satisfies an infinite-dimensional ODE 
\begin{equation}\label{eqn:ODE_inf_dim}
\resizebox{\hsize}{!}{$\frac{d}{dt} \left( \begin{array}{c}
            \ugamma\\
            \ugamma^{\otimes 2} \\
            \ugamma^{\otimes 3} \\
            \vdots \\
            \ugamma^{\otimes N} \\
            \vdots
        \end{array} \right) = 
        \left( \begin{array}{ccccccc}
            A_1^1 & A_2^1 & & & & & \\
            A_1^2 & A_2^2 & A_3^2 & & & &  \\
            & A_2^3 & A_3^3 & A_4^3 & & &\\
            & & \ddots & \ddots & \ddots & &\\
            & & & A_{N-1}^N & A_{N}^{N} & A_{N+1}^N &\\
            & & & & \ddots & \ddots & \ddots 
        \end{array} \right)
        \left( \begin{array}{c}
            \ugamma\\
            \ugamma^{\otimes 2} \\
            \ugamma^{\otimes 3} \\
            \vdots \\
            \ugamma^{\otimes N} \\
            \vdots
        \end{array} \right)
        + 
        \left( \begin{array}{c}
            \Fzerogamma(t)\\
            0 \\
            0 \\
            \vdots \\
            0 \\
            \vdots
        \end{array} \right). $}
\end{equation}
Here $A_{j+1}^j \in \mathbb{C}^{n^j \times n^{j+1}}$, $A_{j}^j \in \mathbb{C}^{n^j \times n^{j}}$, and $A_{j-1}^j \in \mathbb{C}^{n^j \times n^{j-1}}$ are defined as
\begin{align*}
    A_{j+1}^j &= \Ftwogamma(t)\otimes I^{\otimes j-1} + I \otimes \Ftwogamma(t) \otimes I^{\otimes j-2} + \cdots + I^{\otimes j-1} \otimes \Ftwogamma(t), \\
    A_j^j &= \Fonegamma(t)\otimes I^{\otimes j-1} + I \otimes \Fonegamma(t) \otimes I^{\otimes j-2} + \cdots + I^{\otimes j-1} \otimes \Fonegamma(t), \\
    A_{j-1}^j &= \Fzerogamma(t)\otimes I^{\otimes j-1} + I \otimes \Fzerogamma(t) \otimes I^{\otimes j-2} + \cdots + I^{\otimes j-1} \otimes \Fzerogamma(t). 
\end{align*}
Truncating~\cref{eqn:ODE_inf_dim} at a finite order $N$, we obtain a system of linearized ODEs as 
\begin{equation}\label{eqn:ODE_linearized}
    \frac{d \hat{y}}{dt} = A(t) \hat{y} + b(t), 
\end{equation}
where $\hat{y} = [\hat{y}_1;\hat{y}_2;\cdots;\hat{y}_N]$ and
\begin{equation}\label{eq:def_A_b}
     A(t) = 
        \left( \begin{array}{cccccc}
            A_1^1 & A_2^1 & & & & \\
            A_1^2 & A_2^2 & A_3^2 & & & \\
            & A_2^3 & A_3^3 & A_4^3 & &\\
            & & \ddots & \ddots & \ddots &\\
            & & & A_{N-2}^{N-1} & A_{N-1}^{N-1} & A_{N}^{N-1} \\
            & & & & A_{N-1}^N & A_N^N 
        \end{array} \right), \qquad
        b(t) = \left( \begin{array}{c}
            \Fzerogamma(t)\\
            0 \\
            0 \\
            \vdots \\
            0 \\
            0
        \end{array} \right). 
\end{equation}
The initial condition is chosen to be $\hat{y}(0) = [\ugamma(0);\ugamma(0)^{\otimes 2};\cdots;\ugamma(0)^{\otimes N}]$. 
We expect the solution $\hat{y}_j \approx \ugamma^{\otimes j}$ when the truncation order $N$ is sufficiently large, and a technical convergence analysis will be presented next. 

\subsection{Convergence analysis}
We define $\eta \in \mathbb{C}^{\sum^{N}_{j=1} n^j}$ the error vector as 
\begin{equation}\label{eq:def_Carleman_error}
    \eta := [\eta_1; \eta_2; \cdots; \eta_{N}], \qquad \eta_j = \ugamma^{\otimes j} - \hat{y}_j. 
\end{equation}
Then $\eta$ solves the set of linear coupled equations
    \begin{align}
    \frac{d\eta}{dt} = A \eta(t) + \zeta(t), \quad
    \zeta(t) := [0; 0; \dots ;0; A^N_{N+1} \, \ugamma^{\otimes (N+1)}(t)], \quad \eta(0) = [0;\dots;0]. \label{eq:error_evolution}
    \end{align}
The convergence analysis of the Carleman embedding focuses on an upper bound of the norm $\|\eta(t)\|$. 
The idea is to show a fast decay of the error evolution~\cref{eq:error_evolution} under a proper rescaling. 
To this end, we choose the rescaling factor $\gamma$ to be smaller root of the equation 
    \begin{equation}\label{eqn:choice_gamma}
        \frac{\|F_2\|}{\gamma} + \|F_0\|\gamma = -\frac{1+R}{2} \mu(F_1). 
    \end{equation}
    Then we can establish the following two lemmas, whose proofs are given in~\cref{app:Carleman_error_decay}. 

\begin{lemma}\label{lem:matrix_dissipation}
    Suppose the rescaling factor $\gamma$ is chosen to be smaller root of~\cref{eqn:choice_gamma}. 
    Then the coefficient matrix $A(t)$ defined in~\cref{eq:def_A_b} satisfies 
    \begin{equation*}
        A(t) + A(t)^{\dagger} \leq (1-R) \mu(F_1) , \quad \forall~t \in [0,T]. 
    \end{equation*}
\end{lemma}

\begin{lemma}\label{lem:solution_norm_bound}
    Suppose the rescaling factor $\gamma$ is chosen to be smaller root of~\cref{eqn:choice_gamma}. 
    Then the norm of the rescaled solution $\ugamma(t)$ satisfies 
    \begin{equation*}
        \|\ugamma(t)\| \leq \|\ugamma(0)\| < \frac{2R}{1+R}. 
    \end{equation*}
\end{lemma}

An upper bound of the Carleman error vector $\eta$ can be obtained by applying Duhamel's principle to~\cref{eq:error_evolution},~\cref{lem:matrix_dissipation} and~\cref{lem:solution_norm_bound}. 
We formally state the error estimate in the next lemma and provide its proof in~\cref{app:Carleman_error_decay}.

\begin{lemma}\label{lem:Carleman_error_bound}
    Suppose the rescaling factor $\gamma$ is chosen to be the smaller root of~\cref{eqn:choice_gamma}. 
    Then the Carleman error $\eta$ defined in~\cref{eq:def_Carleman_error} can be bounded as 
    \begin{equation*}
        \|\eta(t)\| \leq \frac{ 2 N R }{1-R} \left(\frac{2R}{1+R}\right)^{N}. 
    \end{equation*}
\end{lemma}

\section{Fast-forwarded quantum algorithm}\label{sec:algorithm}

In this section, we show how to design an efficient quantum algorithm for~\cref{eqn:nonlinear_ODE} with fast-forwarded scaling. 
We first specify the quantum access models to the input data, then describe the key steps of our algorithm, and finally present its complexity analysis. 

\subsection{Input models}\label{sec:input_model}

Throughout, we assume coherent oracle access to the coefficient matrices, i.e., time-dependent versions of the block-encodings $O_{F_0}$,$O_{F_{1}}$,$O_{F_2}$ of $F_0(t),F_1(t)$,$F_2(t)$ as 
\begin{equation*}
    (\bra{0}\otimes I) O_{F_i}(\ket{0}\otimes \ket{t} \otimes  I) = \ket{t}
    \frac{F_i(t)}{\alpha_{F_i}}. 
\end{equation*}
Here $\alpha_{F_i}$'s are unified normalization factors such that $\alpha_{F_i} \geq \sup_t \|F_i\|$, and $\ket{t}$ represents a binary encoding of time $t$ which may vary in the algorithm design. 
Notice that such a time-dependent coherent block-encoding model has been widely used in the literature for simulating unitary or non-unitary dynamics with time-dependent coefficients~\cite{LowWiebe2019,BerryCosta2022,FangLinTong2022,AnChildsLin2023}, and, though generally challenging and application dependent, block-encodings of matrices with certain structure and pattern can be efficiently built~\cite{GilyenSuLowEtAl2019,CampsLinVanBeeumenEtAl2024,SunderhaufCampbellCamps2024}. 

For the initial condition, we assume its state preparation oracle $O_u$ such that 
\begin{equation*}
    O_u \ket{0} = \ket{u(0)}.
\end{equation*}
Similar to the block-encoding model, although constructing $O_u$ may be costly for an arbitrary vector, efficient state-preparation circuits are known for certain structured inputs~\cite{GroverRudolph2002,ZhangLiYuan2022,McArdleGilyenBerta2026,RosenkranzBrunnerMarinSanchezEtAl2025}. 
Later, the complexity of our algorithm is measured primarily in terms of oracle query complexity, i.e., the number of calls to $\mathcal{O}_{F_i}$ and $O_u$.

\subsection{Description of the algorithm}\label{sec:description_algorithm}

\paragraph{(i) \emph{Carleman Embedding}}
We first apply the Carleman embedding with proper rescaling discussed in~\cref{sec:Carleman_linearization_procedure} to linearize the original quadratic nonlinear ODE~\cref{eqn:nonlinear_ODE}. 
This yields an enlarged but linear dissipative system given in~\cref{eqn:ODE_linearized}. 

\paragraph{(ii) \emph{Quantum solution of the linearized ODE}}

Following the approach in~\cite{CostaSchleichMoralesBerry2025,JenningsKorzekwaLostaglioEtAl2025} with all the operations controlled by the additional clock register, we construct the time-dependent block-encoding of $A(t)$ as 
\begin{equation*}
    (\bra{0}\otimes I) U_{A}(\ket{0}\otimes \ket{t} \otimes  I) = \ket{t}
    \frac{A(t)}{\alpha_{A}}, 
\end{equation*}
using $\mathcal{O}(1)$ queries to $O_{F_0},O_{F_1},O_{F_2}$. 
The block-encoding factor $\alpha_A$ scales as $\mathcal{O}(N (\gamma \alpha_{F_0} + \alpha_{F_1} + \alpha_{F_2}/\gamma) )$. 
Additionally, following the approach in~\cite{LiuKoldenKroviEtAl2021,CostaSchleichMoralesBerry2025,JenningsKorzekwaLostaglioEtAl2025}, we prepare the initial state $\hat{y}(0)$ using $\mathcal{O}(N)$ queries to $O_u$. 

Then, with these two input models, we apply the fast-forwarded quantum linear ODE algorithm established in~\cite{YangOnwuntaAn2025} to solve the linearized ODE~\cref{eqn:ODE_linearized}. 
The key idea of this algorithm is based on a truncated Duhamel's principle as 
\begin{equation*}
        \hat{y}(T) = \mathcal{T} e^{\int_0^T A(t) dt} \hat{y}(0) + \int_0^T  \mathcal{T} e^{\int_t^T A(s) ds} b(t) dt \approx \int_{T-T_0}^T  \mathcal{T} e^{\int_t^T A(s) ds} b(t) dt. 
\end{equation*}
Thanks to the dissipation of $A(t)$ shown in~\cref{lem:matrix_dissipation}, for sufficiently long time $T$, ignoring the homogeneous term does not notably sacrifice the accuracy, and the time truncation parameter $T_0$ can be independent of $T$ (but depends on the strength of the dissipation as well as the target accuracy). 
Then the algorithm implements each $\mathcal{T} e^{\int_t^T A(s) ds}$ with effective simulation time at most $T_0$ by the linear combination of Hamiltonian simulation approach~\cite{AnChildsLin2023,LowSomma2025}, and then computes the integral by implementing its suitable discretization via the linear combination of unitaries technique~\cite{ChildsWiebe2012}. 
The output of the algorithm is a quantum state $\ket{\hat{y}_{\text{num}}}$, which is an approximation of the normalized final solution $\ket{\hat{y}(T)}$ of the enlarged linear ODE.

\paragraph{(iii) \emph{Post-processing}}

The final step is to extract $\ket{u(T)} = \ket{\ugamma(T)}$ from the enlarged solution of the linearized ODE, which can be realized by simply discarding proper ancilla qubits. 
To understand this, notice that if we ignore the Carleman error, then the exact enlarged state is a tensor product state as 
\begin{equation*}
    \ket{\hat{v}} = \frac{1}{C} [ \ugamma(T); \ugamma(T)^{\otimes 2}; \cdots; \ugamma(T)^{\otimes N} ]  = \ket{\widetilde{v}} \otimes \ket{\ugamma(T)}, 
\end{equation*}
where 
\begin{equation*}
    \ket{\widetilde{v}} = \frac{\|\ugamma(T)\|}{C} [\mathbf{1};  \ugamma(T); \ugamma(T)^{\otimes 2}; \cdots; \ugamma(T)^{\otimes N-1}]. 
\end{equation*}
Therefore, directly discarding all the qubits of $\ket{\widetilde{v}}$ in $\ket{\hat{v}}$ yields the normalized solution $\ket{\ugamma(T)}$ of the original nonlinear ODE. 

In the algorithm, we may also directly discard all the corresponding qubits in $\ket{\hat{y}_{\text{num}}}$. 
The output state is a density matrix $ \rho_{\text{num}} = \Tr_1 ( \ket{\hat{y}_{\text{num}}} \bra{\hat{y}_{\text{num}}}) \approx \Tr_1 ( \ket{ \hat{v} } \bra{ \hat{v} } ) = \ket{\ugamma(T)} \bra{\ugamma(T)} $, where $\Tr_1$ denotes the partial trace operator.

\subsection{Complexity analysis}

The complexity of our algorithm is summarized in the following theorem. 

\begin{theorem}\label{thm:complexity}
    Consider the initial value problem of the nonlinear ODEs in~\cref{eqn:nonlinear_ODE} satisfying the dissipative condition $\mu(F_1) < 0$ and the weakly nonlinear condition $R = \frac{1}{-\mu(F_1)}(\|F_2\|\|u(0)\| + \frac{\|F_0\|}{\|u(0)\|}) < 1$, as discussed in~\cref{sec:setup}. 
    Suppose that we are given access to the time-dependent block-encodings $O_{F_i}$ of the coefficients $F_i(t)$ for $i = 0,1,2$ with block-encoding factor $\alpha_i$ and the state preparation oracle $O_u$ to the initial condition $\ket{u(0)}$ as shown in~\cref{sec:input_model}. 
    Then, for sufficiently large $T>0$ and small $\epsilon > 0$, 
    there exists a quantum algorithm that produces an approximation of $\ket{u(T)}\bra{u(T)}$ with trace distance at most $\epsilon$, using 
    \begin{equation*}
           \widetilde{\mathcal{O}}\left(  \frac{\|u(0)\| + \|F_0\|_{L_1(T-T_0,T)} }{ \|u(T)\|}   \frac{  \gamma \alpha_{F_0} + \alpha_{F_1} + \alpha_{F_2}/\gamma  }{ (1-R)^{5/2} |\mu(F_1)| } \log^{4+o(1)}\left(\frac{1}{\epsilon}\right)\right)
    \end{equation*}
    queries to $O_{F_i}$'s, and 
    \begin{equation*}
         \mathcal{O}\left(\frac{1}{(1-R)^{3/2}} \frac{\|u(0)\| + \|F_0\|_{L_1(T-T_0,T)} }{ \|u(T)\|} \log\left( \frac{|\mu(F_1)|}{\|F_2\|\|u(T)\|\epsilon } \right) \right)
    \end{equation*}
    queries to $O_u$. 
    Here $\|F_0\|_{L_1(T-T_0,T)} = \int_{T-T_0}^T \|F_0(t)\| dt$ where the parameter $T_0 = \widetilde{\mathcal{O}} \left( \frac{1}{(1-R)|\mu(F_1)|} \log\left( \frac{\|F_0\|}{\|u(T)\|\epsilon } \right) \right)$. 
\end{theorem}

\begin{proof}

The proof consists of two parts. 
First, we will specify the tolerated errors as well as parameter choices in each numerical approximation step to make sure that the overall error can be bounded by $\epsilon$. 
Then, we will estimate the overall complexity in the quantum implementation step. 

Let $v(t) = [ \ugamma(t); \ugamma(t)^{\otimes 2}; \cdots; \ugamma(t)^{\otimes N} ]$ and $\ket{v(t)} = v(t)/\|v(t)\|$. 
According to~\cref{lem:Carleman_error_bound}, we can bound the Carleman embedding error as 
\begin{equation*}
    \|\hat{y}(T) - v(T) \| \leq \frac{ 2 N R }{1-R} \left(\frac{2R}{1+R}\right)^{N}. 
\end{equation*}
Notice that 
\begin{equation*}
    \|v(T)\| = \sqrt{ \frac{\|\ugamma (T)\|^2-\|\ugamma (T)\|^{2N+2}}{1-\|\ugamma (T)\|^2} } \ge \|\ugamma(T)\| \geq \frac{2\|F_2\| \|u(T)\| }{ (1+R)|\mu(F_1)| } , 
\end{equation*}
where the last inequality comes from the definition of $\gamma$ in~\cref{eqn:choice_gamma} that $\gamma = \frac{2}{-(1+R)\mu(F_1)} (\|F_2\| + \|F_0\|\gamma^2 ) \geq \frac{2\|F_2\|}{ -(1+R)\mu(F_1) } $. 
Then we can bound the Carleman embedding error in the normalized states as
\begin{align}
    \left\| \ket{\hat{y}(T)} - \ket{v(T)} \right\| & \leq \frac{ 2 \|\hat{y}(T) - v(T)\| }{\|v(T)\|} \nonumber \\
    & \leq 2 N \frac{  R (1+R)  }{1-R} \frac{|\mu(F_1)|}{\|F_2\| \|u(T)\|} \left(\frac{2R}{1+R}\right)^{N} \eqqcolon \epsilon_{\text{carl}}. \label{eq:proof_main_thm_Carleman_error}
\end{align}
After the Carleman embedding, we solve the linear ODE~\cref{eqn:ODE_linearized} by the fast-forwarded quantum linear combination of Hamiltonian simulation algorithm in~\cite{YangOnwuntaAn2025}. 
Let $\ket{\hat{y}_{num}}$ denote the numerical output quantum state of the algorithm such that $\| \ket{\hat{y}_{num}} - \ket{ \hat{y} (T)} \| \leq \epsilon_{\text{num}}$. 
Then, together with~\cref{eq:proof_main_thm_Carleman_error}, we have 
\begin{equation*}
    \| \ket{\hat{y}_{num}} - \ket{ v(T) } \| \leq \epsilon_{\text{carl}} + \epsilon_{\text{num}}. 
\end{equation*}
Let $\rho$ be the final output after discarding the corresponding qubits, i.e., $\rho = \Tr_1 ( \ket{\hat{y}_{num}} \bra{\hat{y}_{num}} )$. 
Then, noticing that $\ket{\ugamma(T)}\bra{\ugamma(T)} = \Tr_1 ( \ket{v(T)} \bra{v(T)} )$ according to the discussion in~\cref{sec:description_algorithm} and using the contractivity of the trace distance (denoted by $D(\cdot,\cdot)$) under the partial trace, 
we can bound the final trace distance error as 
\begin{align*}
    D(\rho,\ket{\ugamma(T)}\bra{\ugamma(T)}) &= D \left( \Tr_1 ( \ket{\hat{y}_{num}} \bra{\hat{y}_{num}} ), \Tr_1 ( \ket{v(T)} \bra{v(T)} ) \right) \\
    & \leq D\left( \ket{\hat{y}_{num}} \bra{\hat{y}_{num}}, \ket{v(T)} \bra{v(T)} \right)\\
    &\leq \left\| \ket{\hat{y}_{num}} - \ket{v(T)}\right\| \leq \epsilon_{\text{carl}} + \epsilon_{\text{num}}. 
\end{align*}
In order to bound the final error by $\epsilon$, it suffices to choose $ \epsilon_{\text{carl}} = \epsilon_{\text{num}} = \epsilon/2 = \mathcal{O}(\epsilon)$, and the corresponding choice of the Carleman truncation order $N$ can be obtained from~\cref{eq:proof_main_thm_Carleman_error} as 
\begin{align}
    N = \mathcal{O}\left( \frac{1}{\log\left( \frac{1+R}{2R} \right)} \log\left( \frac{  R (1+R) |\mu(F_1)| }{(1-R)\|F_2\| \|u(T)\| \epsilon }  \right) \right) =  \widetilde{\mathcal{O}} \left( \frac{1}{1-R} \log\left( \frac{ |\mu(F_1)| }{\|F_2\| \|u(T)\| \epsilon }  \right) \right). \label{eqn:proof_complexity_eq2}
\end{align}

We now estimate the query complexity, which is only from solving the linearized ODE. 
We first focus on the number of queries to the matrix input models. 
According to~\cite[Appendix D.1]{YangOnwuntaAn2025} and~\cref{lem:matrix_dissipation}, as long as $T > \Omega(\frac{1}{\eta} \log(\frac{\|\hat{y}(0)\|  }{ \|\hat{y}(T)\| \epsilon}))$, the query complexity to $U_A$ (which is also the query complexity to $O_{F_i}$'s because $U_A$ can be built by one query to each of $O_{F_i}$ according to  \cite[Appendix~G]{CostaSchleichMoralesBerry2025}) is 
\begin{equation}\label{eqn:proof_complexity_eq3}
    \widetilde{\mathcal{O}}\left(\frac{\|\hat{y}(0)\| + \|\Fzerogamma\|_{L_1(T-T_0,T) }}{\|\hat{y}(T)\|}\frac{\alpha_A}{\eta} \log^{3+o(1)}\left(\frac{1}{\epsilon}\right)\right), 
\end{equation}
where $\eta = \frac{(1-R)|\mu(F_1)|}{2}$, $T_0 = \mathcal{O}( \frac{1}{\eta} \log(\frac{ \|F_{0,\gamma}\| }{ \eta \|\hat{y}(T)\|\epsilon }) )$, and $\alpha_A = \mathcal{O} (N (\gamma \alpha_{F_0} + \alpha_{F_1} + \alpha_{F_2}/\gamma ) ) $.

By the error analysis of Carleman linearization, we have $\displaystyle \|\hat{y}(T)-v(T)\|\le \epsilon \|v(T)\|/2$ for a small $\epsilon$. 
Thus, we have $\displaystyle \|\hat{y}(T)\| \ge (1-\epsilon)\|v(T)\|\ge \frac{1}{2}\|v(T)\|$, and we can estimate $\displaystyle \frac{\|\hat{y}(0)\| + \| \Fzerogamma\|_{L_1(T-T_0,T)}}{\|\hat{y}(T)\|}$ as
\begin{align*}
    \frac{\|\hat{y}(0)\| + \|\Fzerogamma\|_{L_1(T-T_0,T)}}{\|\hat{y}(T)\|} 
    &\le 2\frac{\|v(0)\| +\|{\Fzerogamma}\|_{L_1(T-T_0,T)}}{\|v(T)\|}  \nonumber \\
    & = 2\frac{\gamma\|u(0)\| \sqrt{ \sum_{j=0}^{N-1} (\gamma\|u(0)\|)^{2j} } + \gamma \|F_0\|_{L_1(T-T_0,T)} }{\gamma\|u(T)\| \sqrt{ \sum_{j=0}^{N-1} (\gamma\|u(T)\|)^{2j} }} \nonumber \\
    & \le 2 \sqrt{ \sum_{j=0}^{N-1} (\gamma\|u(0)\|)^{2j} } \frac{\gamma\|u(0)\| + \gamma \|F_0\|_{L_1(T-T_0,T)} }{\gamma\|u(T)\| } \nonumber \\
    &\le 2 \sqrt{ \frac{1}{1-\gamma\|u(0)\|} } \frac{\|u(0)\| + \|F_0\|_{L_1(T-T_0,T)} }{ \|u(T)\|}. 
\end{align*}
From~\cref{lem:solution_norm_bound}, we have $\displaystyle  \sqrt{ \frac{1}{1-\gamma\|u(0)\|} } \leq  \sqrt{ \frac{1}{1-\frac{2R}{1+R} } } \leq \sqrt{ \frac{2}{1-R} }$, so we can further bound the above expression by 
\begin{equation}\label{eqn:proof_complexity_eq1}
    \frac{\|\hat{y}(0)\| + \|\Fzerogamma\|_{L_1(T-T_0,T)}}{\|\hat{y}(T)\|}  \le \mathcal{O}\Big(\frac{1}{\sqrt{1-R}} \frac{\|u(0)\| + \|F_0\|_{L_1(T-T_0,T)} }{ \|u(T)\|} \Big). 
\end{equation}
Similarly, we can also bound the truncation time 
\begin{equation*}
    T_0 = \mathcal{O}\left( \frac{1}{\eta} \log\left(\frac{ \|F_{0,\gamma}\| }{ \eta \|\hat{y}(T)\|\epsilon }\right) \right) = \widetilde{\mathcal{O}} \left( \frac{1}{\eta} \log\left( \frac{\|F_0\|}{\|u(T)\|\epsilon } \right) \right). 
\end{equation*}
Plugging these estimates in~\cref{eqn:proof_complexity_eq1,eqn:proof_complexity_eq2} of parameters back to~\cref{eqn:proof_complexity_eq3}, we can further bound the query complexity to $O_{F_i}$'s as 
\begin{align*}
    & \widetilde{\mathcal{O}}\left(  \frac{1}{\sqrt{1-R}} \frac{\|u(0)\| + \|F_0\|_{L_1(T-T_0,T)} }{ \|u(T)\|}   \frac{ N (\gamma \alpha_{F_0} + \alpha_{F_1} + \alpha_{F_2}/\gamma)  }{ (1-R) |\mu(F_1)| } \log^{3+o(1)}\left(\frac{1}{\epsilon}\right)\right) \nonumber \\
    = ~& \widetilde{\mathcal{O}}\left(  \frac{\|u(0)\| + \|F_0\|_{L_1(T-T_0,T)} }{ \|u(T)\|}   \frac{  \gamma \alpha_{F_0} + \alpha_{F_1} + \alpha_{F_2}/\gamma  }{ (1-R)^{5/2} |\mu(F_1)| } \log^{4+o(1)}\left(\frac{1}{\epsilon}\right)\right). 
\end{align*}

For the state preparation cost, according to~\cite[Appendix D.1]{YangOnwuntaAn2025}, we need $\mathcal{O}(  \frac{\|\hat{y}(0)\| + \|\Fzerogamma\|_{L_1(T-T_0,T)}}{\|\hat{y}(T)\|}  )$ queries to the state preparation oracle of $\hat{y}(0)$, which can be constructed using $\mathcal{O}(N)$ queries to $O_u$. 
So the overall query complexity to $O_u$ is bounded as 
\begin{align*}
    &\mathcal{O}\left( N \frac{\|\hat{y}(0)\| + \|\Fzerogamma\|_{L_1(T-T_0,T)}}{\|\hat{y}(T)\|}  \right) \nonumber \\
    =~& \mathcal{O}\left(\frac{1}{(1-R)^{3/2}} \frac{\|u(0)\| + \|F_0\|_{L_1(T-T_0,T)} }{ \|u(T)\|} \log\left( \frac{|\mu(F_1)|}{\|F_2\|\|u(T)\|\epsilon } \right) \right). 
\end{align*}

\end{proof}

\cref{thm:complexity} shows that the overall complexity of our approach does not explicitly depend on the overall evolution time $T$, while there are two implicit sources of the time dependence through $\|u(T)\|$ and $\|F_0\|_{L_1(T-T_0,T)}$. 
Dependence on the norm $\|u(T)\|$ is unavoidable due to the necessity of post selection in solving linearized ODE~\cite{AnLiuWangEtAl2022}. 
Additionally, $\|u(T)\|$ might be bounded from below independently of $T$ when the source term $F_0(t)$ does not vanish for large $T$. 
The $L^1$ norm $\|F_0\|_{L_1(T-T_0,T)}$ can be bounded independently of $T$ as $\|F_0\|_{L_1(T-T_0,T)} \leq T_0 \|F_0\| = \widetilde{\mathcal{O}} \left( \frac{\|F_0\|}{(1-R)|\mu(F_1)|} \log\left( \frac{1}{\|u(T)\|\epsilon } \right) \right)$.

\section{Numerical results}\label{sec:numerics}

In this section, we provide numerical results to validate our fast-forwarded quantum algorithm for nonlinear ODEs and explore its effectiveness beyond the theoretically guaranteed regime we have established. 
Specifically, we test three scenarios of weakly nonlinear ODEs: dissipative systems we have studied throughout the paper so far, conservative systems, and non-resonant systems. 
For each scenario, we study the convergence of the Carleman embedding by numerically testing
\begin{equation*}
    E_N:=\max_{t\in[0,T]}\|u(t)-\hat{y}_1^{[N]}(t)\|_2
\end{equation*}
for varying truncation order $N$, where $\hat{y}_1^{[N]}(t)$ denotes the first block of the $N$-th truncated Carleman system. 
In all the cases, the reference solutions of the nonlinear ODEs are computed numerically using the DOP853 integrator in \texttt{scipy}. 
We also study the possibility of applying the fast-forwarded quantum ODE algorithm to solve the linearized system, by numerically computing 
\begin{equation*}
    \delta_N:=-\lambda_{\max}\!\left(\frac{A_N+A_N^\dagger}{2}\right)
\end{equation*}
where $A_N$ is the truncated Carleman matrix at order $N$. 
If $\delta_N>0$, then the truncated lifted linear system is dissipative and can be solved with cost independently of $T$ as discussed earlier.

\subsection{Dissipative systems}

\subsubsection{Setting}
We consider~\cref{eqn:nonlinear_ODE} up to the final time $T = 1$ with
\begin{equation*}
    F_1=
\begin{bmatrix}
-1 & 0\\
0 & -2
\end{bmatrix},
\quad
F_2=f_2
\begin{bmatrix}
1 & 2 & 3 & 4\\
5 & 6 & 7 & 8
\end{bmatrix},
\quad
F_0=f_0
\begin{bmatrix}
\frac12\\
1
\end{bmatrix},
\quad
u(0) =
\begin{bmatrix}
1\\
1
\end{bmatrix}. 
\end{equation*}
The parameters are chosen as 
\begin{equation*}
    (f_2,f_0)\in
\left\{
(0.02,0.2),\,
(0.05,0),\,
(0.05,0.2),\,
(0.02,0),\,
(0.03,0.1),\,
(0.07,0.2)
\right\},
\end{equation*}
with the corresponding values
\begin{equation*}
    R\approx
\{0.56,\ 1.01,\ 1.16,\ 0.40,\ 0.68,\ 1.57\}.
\end{equation*}
Hence both $R < 1$ and $R > 1$ cases are included.

We emphasize that the numerical experiments in this subsection are performed directly on the original systems without rescaling. In particular, no rescaling transformation is applied to the lifted linear systems in the implementation. For each truncation order $N$, the Carleman approximation is obtained by solving the finite-dimensional truncated linear system constructed from the original coefficients $F_0,F_1,F_2$.

\subsubsection{Numerical observations}
Figure~\ref{fig:model1_error} shows the decay of the error $E_N$ as $N$ increases. 
One observes that, for all tested parameter choices, the truncation error decreases monotonically with respect to the Carleman truncation order. 
In particular, for the cases with $R<1$, the convergence is very rapid. 
The error drops by several orders of magnitude as $N$ increases from $1$ to $10$. 
This is fully consistent with the theoretical prediction.

More interestingly, the cases with $R>1$ also exhibit clear numerical convergence. In particular, even when $R\approx 1.01$, $1.16$, or $1.57$, the error still decreases as $N$ grows, although the decay rate becomes significantly slower than in the cases with $R<1$. This suggests that the condition $R<1$ should be viewed as a sufficient condition for Carleman convergence, but not as a sharp necessary threshold for this model.
Additionally, our results also suggest that the rescaling step seems not a necessary requirement for the convergence of the Carleman embedding, as our numerical tests are performed directly on the original equations with initial vector norm larger than $1$. 

Figure~\ref{fig:model1_dissipativity} displays the dissipativity margin $\delta_N$ for the same set of parameters. We observe that $\delta_N$ remains positive for all tested truncation orders and for all tested parameter choices. Hence, in every convergent case considered here (regardless of whether $R < 1$ or $R > 1$), the truncated Carleman linearized system is dissipative at the lifted level, and thus can be solved by the fast-forwarded quantum ODE algorithm. 
Moreover, the values of $\delta_N$ stay uniformly bounded away from zero, indicating that the dissipative structure is stable with respect to the truncation order.

\begin{figure}[t]
    \centering
    \includegraphics[width=0.75\linewidth]{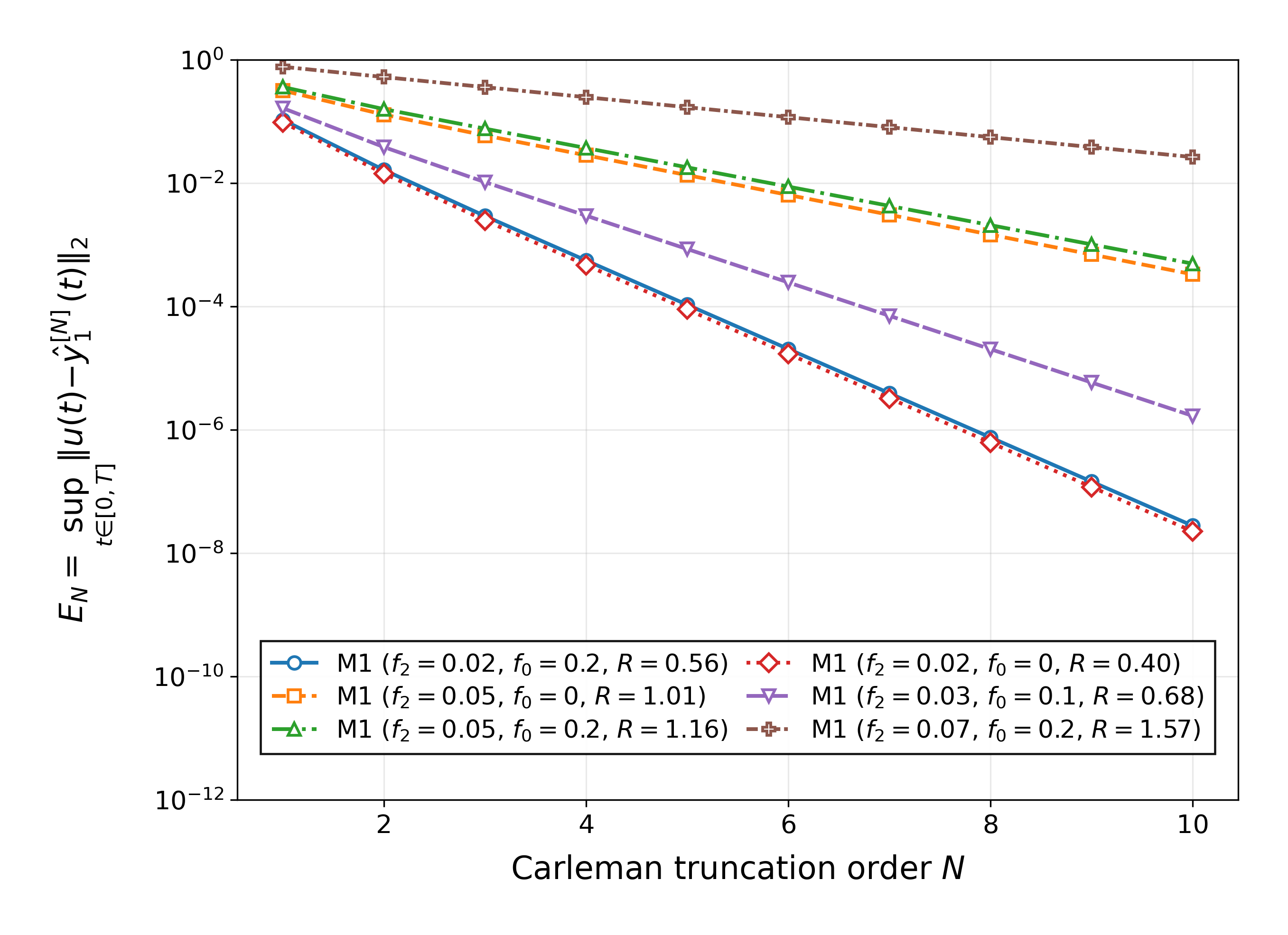}
    \caption{Error decay $E_N=\max_{t\in[0,T]}\|u(t)-\hat{y}_1^{[N]}(t)\|_2$ for dissipative systems under different choices of $(f_2,f_0)$. The plots show that Carleman convergence persists not only in the regime $R<1$, but also for several cases with $R>1$.}
    \label{fig:model1_error}
\end{figure}

\begin{figure}[t]
    \centering
    \includegraphics[width=0.75\linewidth]{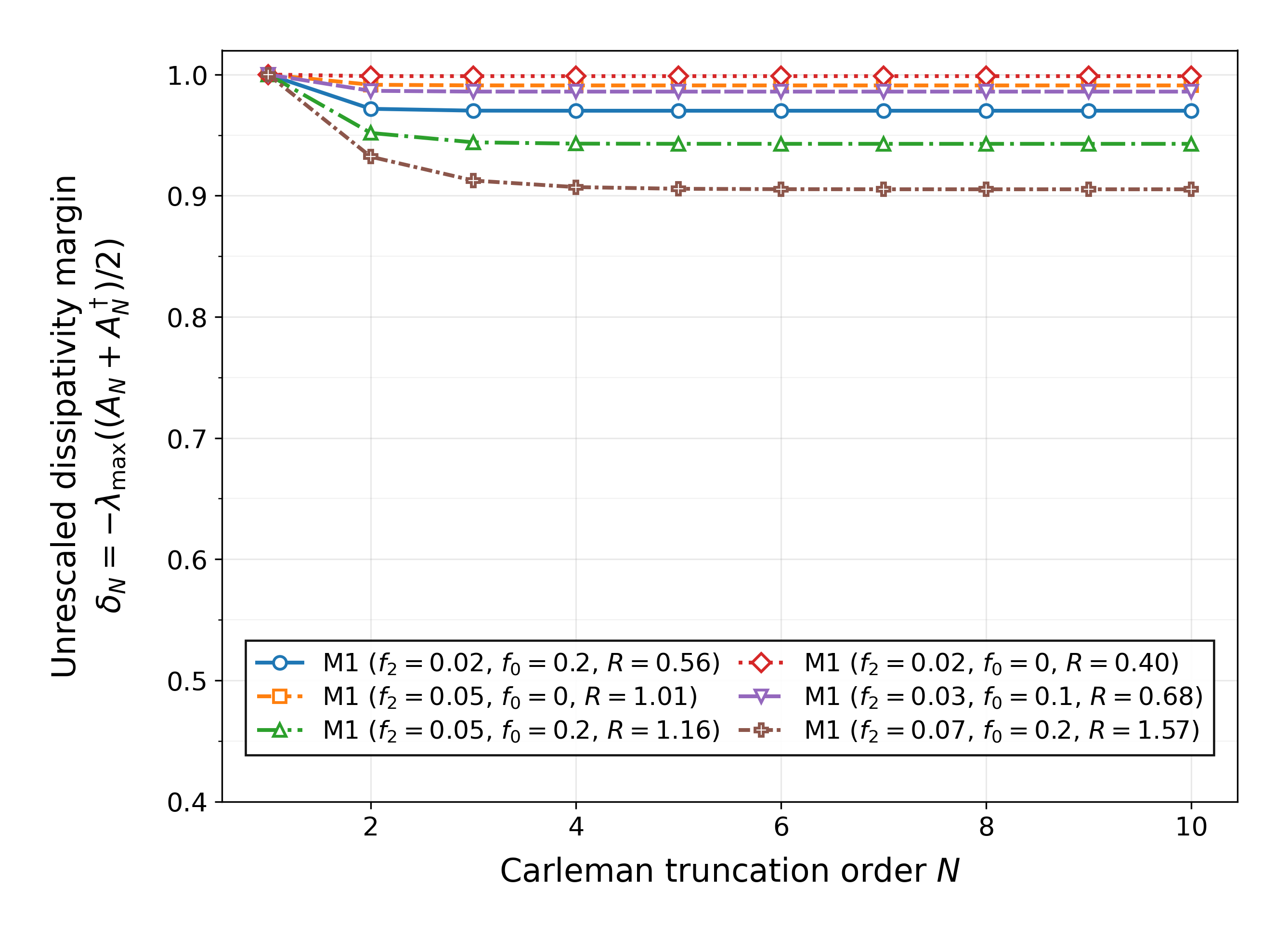}
    \caption{Dissipativity margin $\delta_N=-\lambda_{\max}((A_N+A_N^\dagger)/2)$ for dissipative systems. 
    }
    \label{fig:model1_dissipativity}
\end{figure}

\subsection{Conservative systems}

\subsubsection{Setting}
We consider the homogeneous case of~\cref{eqn:nonlinear_ODE} (i.e., without the $F_0(t)$ term) up to time $T = 3$ 
with
\begin{equation*}
    F_1=
\begin{bmatrix}
0 & 0\\
0 & -1
\end{bmatrix},
\qquad
F_2=
\begin{bmatrix}
0 & 0 & 0 & 0\\
a & 0 & 0 & -b
\end{bmatrix},
\qquad
u(0)=
\begin{bmatrix}
0.2\\
0.1
\end{bmatrix},
\end{equation*}
where $a=0.2, b=1.0$. 
Notice that the first entry of $u$ remains unchanged, and this is why such a system is called (partially) conservative system. 
When the nonlinearity is weak, the convergence of the Carleman embedding for conservative systems has been established in~\cite{JenningsKorzekwaLostaglioEtAl2025}. 

For the numerical reference solution, we again use the DOP853 integrator in \texttt{scipy}. 
In this model, the DOP853 reference can also be checked against the analytic solution. 
For the parameters used here, we obtain $\sup_t\|u_{\mathrm{DOP853}}(t)-u_{\mathrm{exact}}(t)\|_2 = 1.026\times 10^{-11}$,
which confirms that the DOP853 solution is sufficiently accurate as the reference solution for the Carleman truncation error.
\subsubsection{Numerical observations.}
Figure~\ref{fig:model2_error} shows the truncation error $E_N$ as a function of $N$. 
The error decreases rapidly with the truncation order, improving by several orders of magnitude already at small $N$, and reaching the reference-solver accuracy floor by $N=9$ or $N=10$. 
This indicates strong numerical convergence of the Carleman approximation for this model.

However, the dissipativity behavior is qualitatively different. 
As shown in Figure~\ref{fig:model2_dissipativity}, $\delta_1=0$, $\delta_N<0$ for all $N\ge 2$, and $\delta_N$ becomes more negative as $N$ increases. 
Thus, although the Carleman truncation accurately approximates the nonlinear solution, the associated truncated lifted linear system is not dissipative. 
Consequently, the fast-forwarded quantum ODE solver for dissipative linear systems does not apply to this example.

This conservative model therefore illustrates that Carleman convergence and dissipativity of the lifted system are distinct properties. 
The rapid decay of the truncation error reflects accurate approximation of the nonlinear trajectory, but it should not be interpreted as decay of the lifted linear flow itself. 

\begin{figure}[t]
    \centering
    \includegraphics[width=0.75\linewidth]{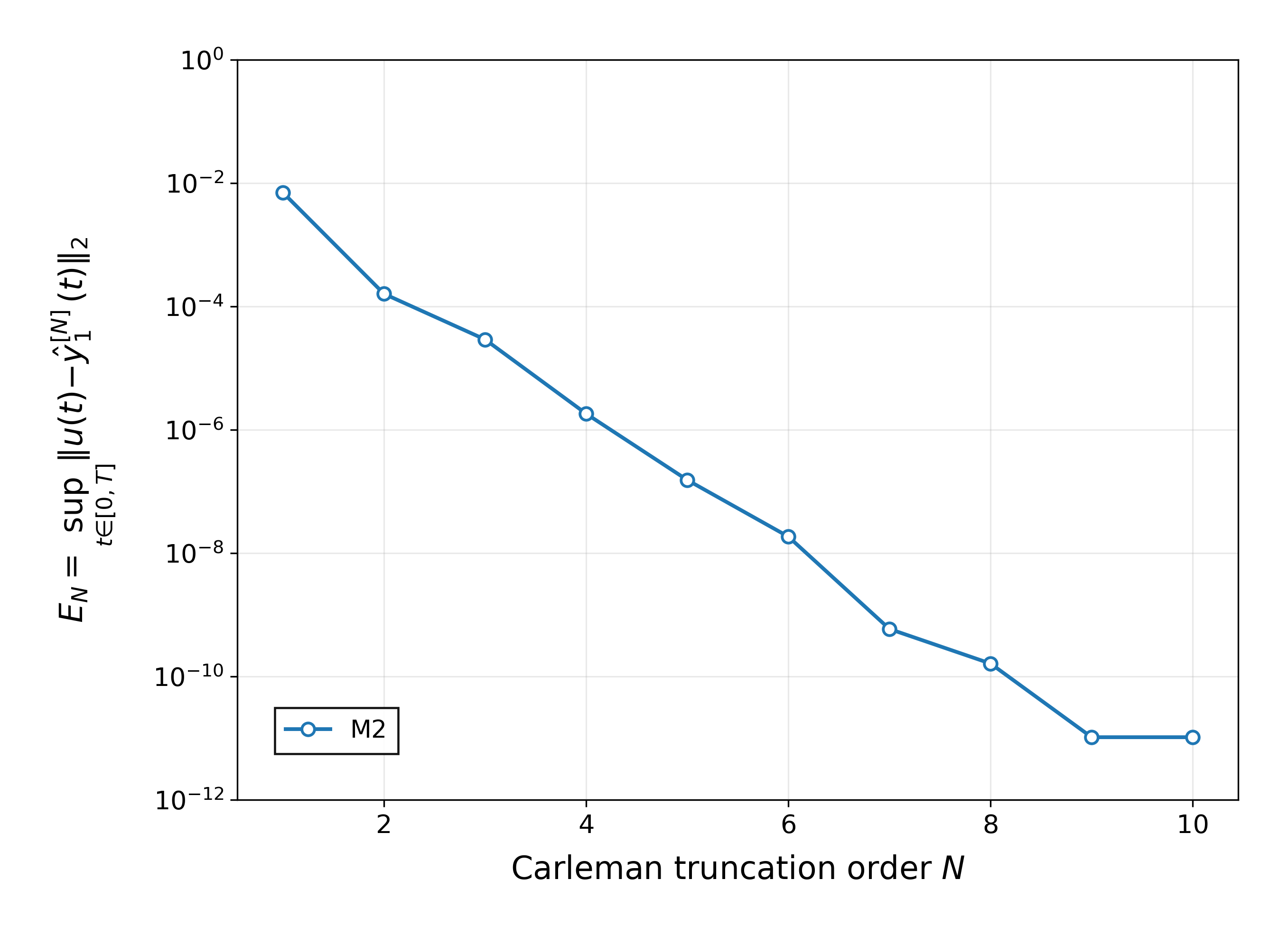}
    \caption{Error decay $E_N=\max_{t\in[0,T]}\|u(t)-\hat{y}_1^{[N]}(t)\|_2$ for Model~2. The truncation error decreases rapidly as the Carleman truncation order increases.}
    \label{fig:model2_error}
\end{figure}

\begin{figure}[t]
    \centering
    \includegraphics[width=0.75\linewidth]{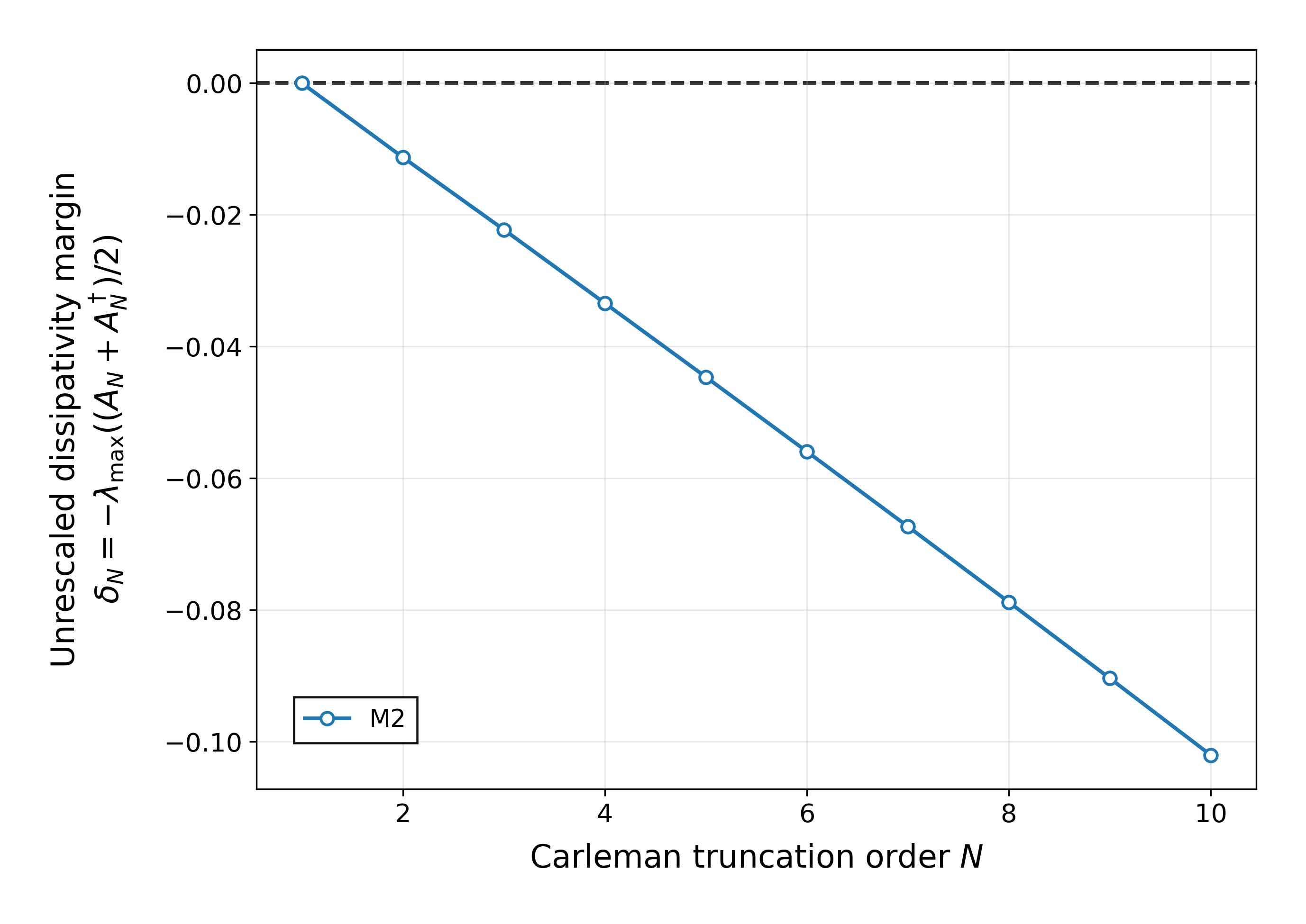}
    \caption{Dissipativity margin $\delta_N=-\lambda_{\max}((A_N+A_N^\dagger)/2)$ for Model~2. The margin is non-positive, showing that the truncated lifted linear system is not dissipative.}
    \label{fig:model2_dissipativity}
\end{figure}

\subsection{Non-resonant systems}

\subsubsection{Setting}
We consider the homogeneous case of~\cref{eqn:nonlinear_ODE} up to time $T = 3$ with 
\begin{equation*}
F_1=
\begin{bmatrix}
-i & 0\\
0 & -2
\end{bmatrix},
\qquad
F_2=f_2
\begin{bmatrix}
0 & 1 & 0 & 0\\
1 & 0 & 0 & -0.5
\end{bmatrix},
\qquad
u(0)=
\begin{bmatrix}
1\\
0.3
\end{bmatrix}. 
\end{equation*}
We vary $f_2\in\{1.0,\,1.5,\,0.3,\,0.5,\,0.4,\,1.2,\,0.1,\,0.2,\,0.09\}$, 
with corresponding values 
\begin{equation*}
    R_\Delta\approx\{14.61,\,29.61,\,3.20,\,5.71,\,4.40,\,19.73,\,1.04,\,2.09,\,0.94\}. 
\end{equation*}
For this non-resonant homogeneous setting, the reported quantity follows the definition in~\cite{JenningsKorzekwaLostaglioEtAl2025}. 
In the notation of that work, one diagonalizes the linear part as $F_1=Q\Lambda Q^{-1}$,
and writes the dynamics in the eigenbasis of $F_1$. 
The transformed quadratic coefficient is $\widetilde F_2=Q^{-1}F_2Q^{\otimes 2}$.
The corresponding non-resonant $R$-number is defined by
\begin{equation*}
    R_\Delta
    =
    \frac{
    8s\|\widetilde F_2\|_2
    \|\widetilde u_{\max}\|_2
    }{
    \Delta(F_1)
    },
    \qquad
    \|\widetilde u_{\max}\|_2
    :=
    \max_{t\in[0,T]}\|Q^{-1}u(t)\|_2,
\end{equation*}
where $s$ denotes the maximum column sparsity of $\widetilde F_2$. 
In the present example, $F_1$ is already diagonal, so $Q=I$, and therefore $\widetilde F_2=F_2, \widetilde u(t)=u(t)$. 
Moreover,
\begin{equation*}
    \Delta(F_1)
    =
    \min_i\inf_{|\alpha|\ge 2}
    \frac{
    \left|\lambda_i-\sum_j\alpha_j\lambda_j\right|
    }{
    |\alpha|-1
    }.
\end{equation*}
Here $\{\lambda_i\}$ are the eigenvalues of $F_1$. Under the non-resonance assumption, $R_\Delta<1$ corresponds to the guaranteed convergence regime.

\subsubsection{Numerical observations}
The error curves in Figure~\ref{fig:model3_error} reveal a clear dependence on the nonlinear strength $R_\Delta$. 
These observations show that the condition $R_\Delta<1$ is again sufficient but not necessary for convergence of Carleman embedding. 
Indeed, the case $R_\Delta\approx 0.94$ converges strongly, as predicted by theory, but several cases with
$R_\Delta>1$ also still converge numerically, including values around $1.04$, $2.09$, $3.20$, $4.40$, and $5.71$. 
However, once $R_\Delta$ becomes sufficiently large, the truncation quality deteriorates substantially and eventually fails in the tested range.

The dissipative behavior of the linearized system is shown in Figure~\ref{fig:model3_dissipativity}. In sharp contrast with the first scenario, here we observe that
$\delta_1=0$ and $\delta_N<0, N\ge 2$
for all tested values of $f_2$. Moreover, the negativity becomes stronger both as $N$ increases and as $f_2$ becomes larger. 
Hence the truncated lifted linear system fails to be dissipative for all tested cases. 
This shows that, in the non-resonant setting, convergence of the Carleman embedding can occur without decay of the truncated lifted flow. 
Therefore, for this model, non-resonance helps explain the observed approximation behavior, but it does not guarantee dissipativity of the linearized system.

\begin{figure}[t]
    \centering
    \includegraphics[width=0.75\linewidth]{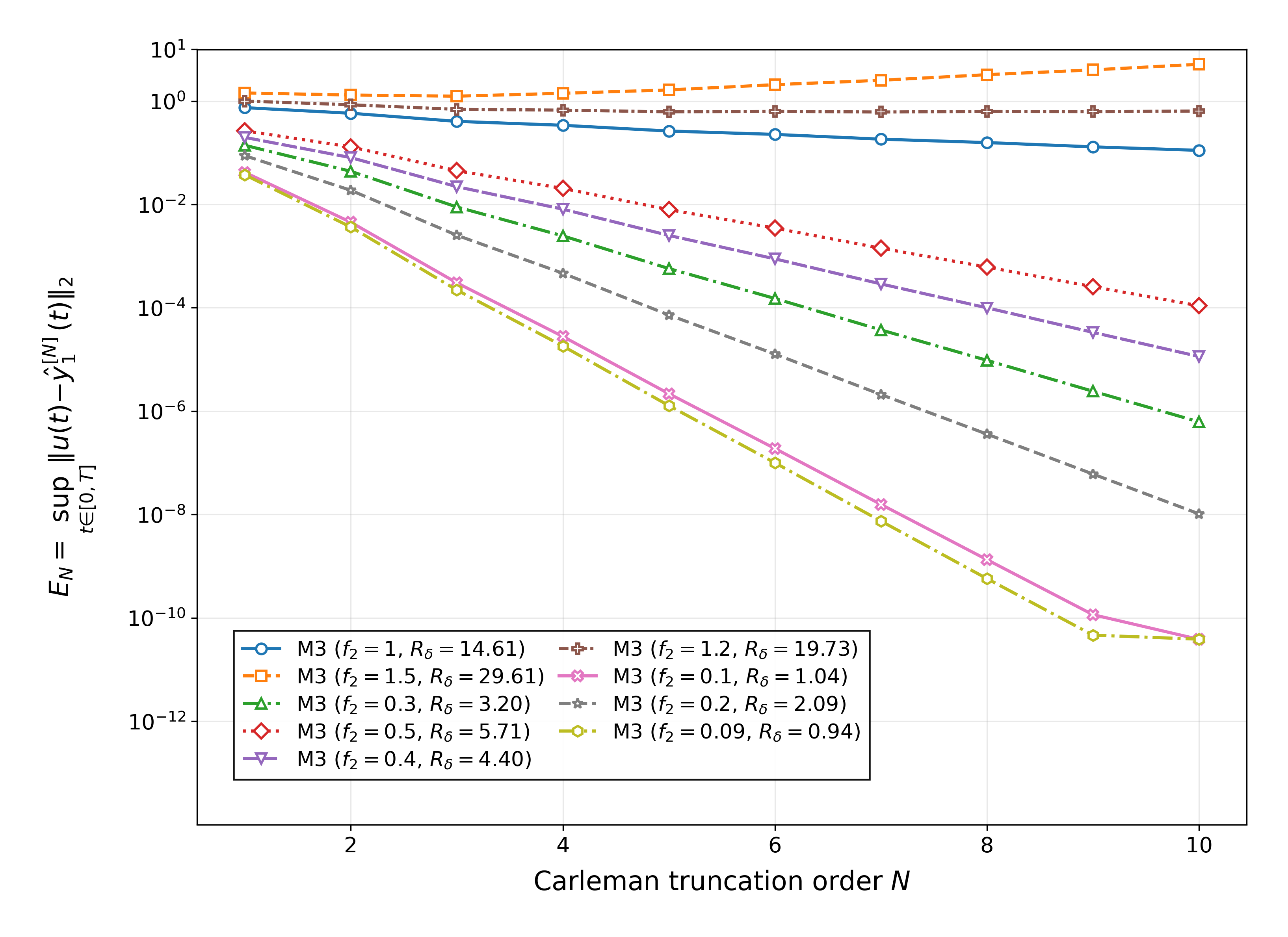}
    \caption{Error decay $E_N=\max_{t\in[0,T]}\|u(t)-\hat{y}_1^{[N]}(t)\|_2$ for Model~3 under different values of $f_2$. The results show good convergence for small and moderate nonlinear strength, while the approximation deteriorates for larger $f_2$.}
    \label{fig:model3_error}
\end{figure}

\begin{figure}[t]
    \centering
    \includegraphics[width=0.75\linewidth]{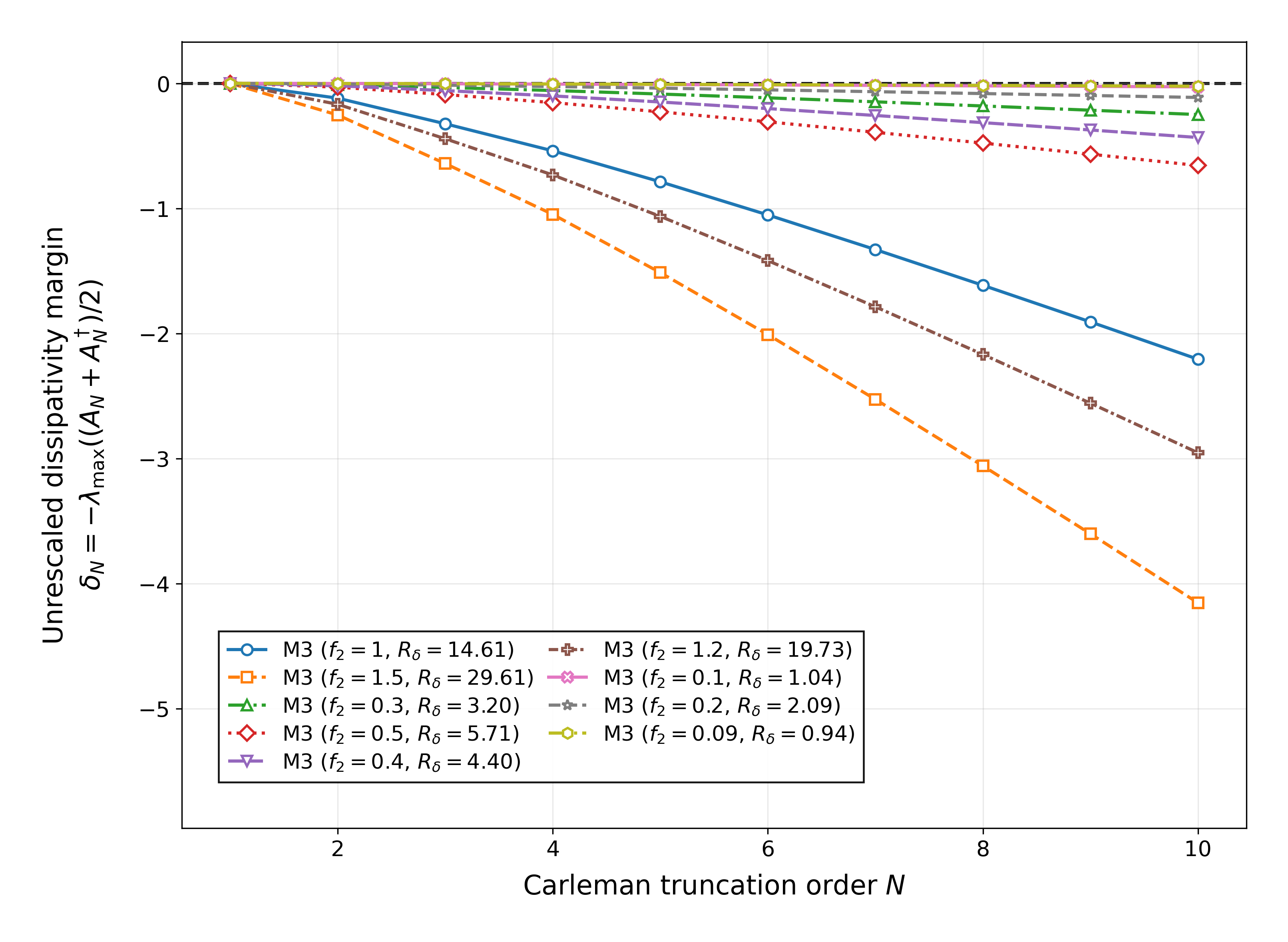}
    \caption{Dissipativity margin $\delta_N=-\lambda_{\max}((A_N+A_N^\dagger)/2)$ for Model~3. The margin is non-positive for all tested cases, indicating that the truncated lifted linear system is not dissipative.}
    \label{fig:model3_dissipativity}
\end{figure}

\subsubsection{Implication on fast-forwarding quantum algorithms}

As shown by our numerical results, the linearized ODEs of the non-resonant systems fail to be dissipative even if the Carleman embedding converges. 
Therefore, the fast-forwarded quantum algorithm described in the previous parts of this work does not apply. 
However, we would like to remark that fast-forwarding is still possible with a different quantum algorithm based on contour integral~\cite{TakahiraOhashiSogabeEtAl2020,TakahiraOhashiSogabeEtAl2021,JiangAn2026}, under the additional condition that $F_1$ and $F_2$ are time-independent and all the eigenvalues of $F_1$ have negative real parts. 
Below we describe the heuristics of the approach. 

Let $A$ denote the coefficient matrix of the linearized ODE. 
By Cauchy's integral theorem, we have $e^{AT} = \frac{1}{2\pi i} \oint_{\Gamma} e^{zT} (zI-A)^{-1} dz$. 
Here $\Gamma$ is a closed curve encompassing all the eigenvalues of $A$. 
Then $e^{AT}$ can be implemented as follows: we numerically discretize the contour integral to approximate $e^{AT}$ as a linear combination of $(zI-A)^{-1}$, implement each $(zI-A)^{-1}$ by quantum linear system solvers such as quantum singular value transformation~\cite{GilyenSuLowEtAl2019,ChakrabortyGilyenJeffery2018}, and compute the summation by the linear combination of unitaries (LCU) technique~\cite{ChildsWiebe2012}. 

When the nonlinear ODE is homogeneous, the coefficient matrix $A$ of the linearized system in~\cref{eqn:ODE_linearized} is block upper triangular. 
Furthermore, its eigenvalues are in the form of $\operatorname{spec}(A)=
\bigcup_{j=1}^{N}
\left\{\lambda_{k_1}+\cdots+\lambda_{k_j}:
\lambda_{k_\ell}\in\operatorname{spec}(F_1)\right\}$, which have negative real part. 
Thus we may choose the contour $\Gamma = \Gamma_1 \cup \Gamma_2$ where $\Gamma_1 = [-i (N \|A\|+1),i(N\|A\|+1) ]$ and $\Gamma_2 = \{z: |z| = N\|A\|+1,~\Re z \leq 0 \}$. 

The complexity of implementing $e^{AT}$ based on the contour integral consists of two parts. 
The first is the cost of implementing $(zI-A)^{-1}$, which depends on the condition number of the matrix $zI-A$ and thus independent of the evolution time. 
The second is the additional cost induced by LCU, which is roughly $\mathcal{O}\left( \oint_{\Gamma} |e^{zT}| \|(zI-A)^{-1}\| dz \right) $. 
This is again independent of the evolution time $T$, as $|e^{zT}| \leq 1$ for all $z \in \Gamma$. 
Therefore, the overall complexity does not explicitly depend on the evolution time, which achieves the fast-forwarding scaling.

\section{Conclusion}\label{sec:concludion}

In this paper, we propose a quantum algorithm for solving weakly nonlinear dissipative ODEs with complexity independent of the evolution time. 
Our algorithm is a combination of Carleman embedding for time-dependent coefficient matrices and the linear combination of Hamiltonian simulation algorithm for linearized ODEs with fast-forwarded scaling. 
We show the convergence of Carleman embedding for time-dependent coefficient matrices, rigorously analyze the performance of our approach, and simplifies the post-selection step in the algorithm. 
We also numerically study weakly nonlinear ODEs beyond the dissipative case and identify possibility for practical fast-forwarding. 

A natural future direction would be fast-forwarding quantum algorithms for broader types of nonlinear ODEs. 
We have briefly sketched the idea of applying the contour-integral-based quantum algorithm to fast-forward weakly nonlinear homogeneous non-resonant ODEs with time-independent coefficients. 
It is interesting to establish a complete rigorous performance analysis of such an approach, as well as extend it to inhomogeneous case. 
Extending this idea to non-resonant ODEs with time-dependent coefficients may require new techniques, as the Carleman convergence has not been shown in this case, and the contour-integral-based quantum algorithm in principle only applies to functions of time-independent matrices. 

\section*{Acknowledgments}

DA acknowledges funding from Quantum Science and Technology - National Science and Technology Major Project via Project 2024ZD0301900, and the support by The Fundamental Research Funds for the Central Universities, Peking University.

\bibliographystyle{unsrt}
\bibliography{ref}

@article{ZhangLiYuan2022,
  author  = {Zhang, Xiao-Ming and Li, Tongyang and Yuan, Xiao},
  title   = {Quantum state preparation with optimal circuit depth: Implementations and applications},
  journal = {Physical Review Letters},
  volume  = {129},
  number  = {23},
  pages   = {230504},
  year    = {2022},
  doi     = {10.1103/PhysRevLett.129.230504}
}

@article{McArdleGilyenBerta2026,
  author  = {McArdle, Sam and Gily{\'e}n, Andr{\'a}s and Berta, Mario},
  title   = {Quantum state preparation without coherent arithmetic},
  journal = {Physical Review Letters},
  volume  = {136},
  pages   = {240603},
  year    = {2026},
  doi     = {10.1103/ntvs-c48s}
}

@article{RosenkranzBrunnerMarinSanchezEtAl2025,
  author  = {Rosenkranz, Matthias and Brunner, Eric and Marin-Sanchez, Gabriel
             and Fitzpatrick, Nathan and Dilkes, Silas and Tang, Yao and
             Kikuchi, Yuta and Benedetti, Marcello},
  title   = {Quantum state preparation for multivariate functions},
  journal = {Quantum},
  volume  = {9},
  pages   = {1703},
  year    = {2025},
  doi     = {10.22331/q-2025-04-09-1703}
}

@article{CampsLinVanBeeumenEtAl2024,
  author  = {Camps, Daan and Lin, Lin and Van Beeumen, Roel and Yang, Chao},
  title   = {Explicit quantum circuits for block encodings of certain sparse matrices},
  journal = {SIAM Journal on Matrix Analysis and Applications},
  volume  = {45},
  number  = {1},
  pages   = {801--827},
  year    = {2024},
  doi     = {10.1137/22M1484298}
}

@article{SunderhaufCampbellCamps2024,
  author  = {S{\"u}nderhauf, Christoph and Campbell, Earl T. and Camps, Joan},
  title   = {Block-encoding structured matrices for data input in quantum computing},
  journal = {Quantum},
  volume  = {8},
  pages   = {1226},
  year    = {2024},
  doi     = {10.22331/q-2024-01-11-1226}
}

@article{HuangKuengPreskill2020,
  author  = {Huang, Hsin-Yuan and Kueng, Richard and Preskill, John},
  title   = {Predicting many properties of a quantum system from very few measurements},
  journal = {Nature Physics},
  volume  = {16},
  pages   = {1050--1057},
  year    = {2020},
  doi     = {10.1038/s41567-020-0932-7}
}

@article{ElbenFlammiaHuangEtAl2023,
  author  = {Elben, Andreas and Flammia, Steven T. and Huang, Hsin-Yuan and
             Kueng, Richard and Preskill, John and Vermersch, Beno{\^i}t and
             Zoller, Peter},
  title   = {The randomized measurement toolbox},
  journal = {Nature Reviews Physics},
  volume  = {5},
  pages   = {9--24},
  year    = {2023},
  doi     = {10.1038/s42254-022-00535-2}
}

@misc{AminiZhengSunEtAl2022,
      title={Carleman Linearization of Nonlinear Systems and Its Finite-Section Approximations}, 
      author={Arash Amini and Cong Zheng and Qiyu Sun and Nader Motee},
      year={2022},
      eprint={2207.07755},
      archivePrefix={arXiv},
      primaryClass={math.DS},
      url={https://arxiv.org/abs/2207.07755}, 
}

@misc{ForetsPouly2017,
      title={Explicit Error Bounds for Carleman Linearization}, 
      author={Marcelo Forets and Amaury Pouly},
      year={2017},
      eprint={1711.02552},
      archivePrefix={arXiv},
      primaryClass={math.NA},
      url={https://arxiv.org/abs/1711.02552}, 
}

@misc{HuJin2026,
      title={Quantum Simulation of Non-Unitary Dynamics via Amplitude-Phase Separation}, 
      author={Qitong Hu and Shi Jin},
      year={2026},
      eprint={2602.09575},
      archivePrefix={arXiv},
      primaryClass={quant-ph},
      url={https://arxiv.org/abs/2602.09575}, 
}

@article{Li2026,
  title = {From Linear Differential Equations to Unitaries: A Moment-Matching Dilation Framework with Near-Optimal Quantum Algorithms},
  author = {Li, Xiantao},
  journal = {PRX Quantum},
  volume = {7},
  issue = {2},
  pages = {020350},
  numpages = {20},
  year = {2026},
  month = {Jun},
  publisher = {American Physical Society},
  doi = {10.1103/xm61-ytf7},
  url = {https://link.aps.org/doi/10.1103/xm61-ytf7}
}

@article{TakahiraOhashiSogabeEtAl2021,
   title={Quantum algorithms based on the block-encoding framework for matrix functions by contour integrals},
   volume={20},
   ISSN={1533-7146},
   url={https://doi.org/10.26421/QIC22.11-12-4},
   DOI={10.26421/QIC22.11-12-4},
   number={11--12},
   journal={Quantum Information and Computation},
   publisher={Rinton Press},
   author={Takahira, Souichi and Ohashi, Asuka and Sogabe, Tomohiro and Usuda, Tsuyoshi S.},
   year={2021},
    pages={965–979} }

@misc{JiangAn2026,
      title={Contour-integral based quantum eigenvalue transformation: analysis and applications}, 
      author={Shan Jiang and Dong An},
      year={2026},
      eprint={2601.11959},
      archivePrefix={arXiv},
      primaryClass={quant-ph},
      url={https://arxiv.org/abs/2601.11959}, 
}

@misc{LowSomma2025,
      title={Optimal quantum simulation of linear non-unitary dynamics}, 
      author={Guang Hao Low and Rolando D. Somma},
      year={2025},
      eprint={2508.19238},
      archivePrefix={arXiv},
      primaryClass={quant-ph},
      url={https://arxiv.org/abs/2508.19238}, 
}

@article{ShangGuoAnZhao2025,
   title={Designing a Nearly Optimal Quantum Algorithm for Linear Differential Equations via Lindbladians},
   volume={135},
   ISSN={1079-7114},
   url={http://dx.doi.org/10.1103/cvl9-97qg},
   DOI={10.1103/cvl9-97qg},
   number={12},
   journal={Physical Review Letters},
   publisher={American Physical Society (APS)},
   author={Shang, Zhong-Xia and Guo, Naixu and An, Dong and Zhao, Qi},
   year={2025},
   month=Sept }

@misc{JinLiu2026HJB,
      title={Quantum algorithms for viscosity solutions to nonlinear Hamilton-Jacobi equations based on an entropy penalisation method}, 
      author={Shi Jin and Nana Liu},
      year={2026},
      eprint={2512.07919},
      archivePrefix={arXiv},
      primaryClass={quant-ph},
      url={https://arxiv.org/abs/2512.07919}, 
}

@misc{JinLiuMedvidovaYuan2026,
      title={Quantum algorithms for Young measures: applications to nonlinear partial differential equations}, 
      author={Shi Jin and Nana Liu and Maria Lukacova-Medvidova and Yuhuan Yuan},
      year={2026},
      eprint={2604.11825},
      archivePrefix={arXiv},
      primaryClass={quant-ph},
      url={https://arxiv.org/abs/2604.11825}, 
}

@misc{JenningsKorzekwaLostaglioWang2026,
      title={Quantum Koopman Algorithms}, 
      author={David Jennings and Kamil Korzekwa and Matteo Lostaglio and Guoming Wang},
      year={2026},
      eprint={2605.19054},
      archivePrefix={arXiv},
      primaryClass={quant-ph},
      url={https://arxiv.org/abs/2605.19054}, 
}

@misc{katzMuraleedharanAlase2025,
      title={Efficient quantum algorithm for solving differential equations with Fourier nonlinearity via Koopman linearization}, 
      author={Judd Katz and Gopikrishnan Muraleedharan and Abhijeet Alase},
      year={2025},
      eprint={2512.06488},
      archivePrefix={arXiv},
      primaryClass={quant-ph},
      url={https://arxiv.org/abs/2512.06488}, 
}

@article{Joseph2020,
   title={Koopman–von Neumann approach to quantum simulation of nonlinear classical dynamics},
   volume={2},
   ISSN={2643-1564},
   url={http://dx.doi.org/10.1103/PhysRevResearch.2.043102},
   DOI={10.1103/physrevresearch.2.043102},
   number={4},
   journal={Physical Review Research},
   publisher={American Physical Society (APS)},
   author={Joseph, Ilon},
   year={2020},
   month=Oct }

@misc{WangJiaVeerapaneniDing2026,
      title={Quantum Algorithms for Nonlinear Differential Equations via Pivot-Shifted Carleman Linearization}, 
      author={Ke Wang and Zikang Jia and Shravan Veerapaneni and Zhiyan Ding},
      year={2026},
      eprint={2605.20071},
      archivePrefix={arXiv},
      primaryClass={quant-ph},
      url={https://arxiv.org/abs/2605.20071}, 
}

@article{JinLiu2024,
   title={Quantum algorithms for nonlinear partial differential equations},
   volume={194},
   ISSN={0007-4497},
   url={http://dx.doi.org/10.1016/j.bulsci.2024.103457},
   DOI={10.1016/j.bulsci.2024.103457},
   journal={Bulletin des Sciences Mathématiques},
   publisher={Elsevier BV},
   author={Jin, Shi and Liu, Nana},
   year={2024},
   month=Sept, pages={103457} }

@article{CostaSchleichMoralesBerry2025,
   title={Further improving quantum algorithms for nonlinear differential equations via higher-order methods and rescaling},
   volume={11},
   ISSN={2056-6387},
   url={http://dx.doi.org/10.1038/s41534-025-01084-z},
   DOI={10.1038/s41534-025-01084-z},
   number={1},
   journal={npj Quantum Information},
   publisher={Springer Science and Business Media LLC},
   author={Costa, Pedro C. S. and Schleich, Philipp and Morales, Mauro E. S. and Berry, Dominic W.},
   year={2025},
   month=aug }

@misc{JenningsKorzekwaLostaglioEtAl2025,
      title={Quantum algorithms for general nonlinear dynamics based on the Carleman embedding}, 
      author={David Jennings and Kamil Korzekwa and Matteo Lostaglio and Andrew T Sornborger and Yigit Subasi and Guoming Wang},
      year={2025},
      eprint={2509.07155},
      archivePrefix={arXiv},
      primaryClass={quant-ph},
      url={https://arxiv.org/abs/2509.07155}, 
}

@misc{WuWangLi2025,
      title={Quantum Algorithms for Nonlinear Dynamics: Revisiting Carleman Linearization with No Dissipative Conditions}, 
      author={Hsuan-Cheng Wu and Jingyao Wang and Xiantao Li},
      year={2025},
      eprint={2405.12714},
      archivePrefix={arXiv},
      primaryClass={quant-ph},
      url={https://arxiv.org/abs/2405.12714}, 
}

@misc{YangOnwuntaAn2025,
      title={Quantum Differential Equation Solvers with Low State Preparation Cost: Eliminating the Time Dependence in Dissipative Equations}, 
      author={Gengzhi Yang and Akwum Onwunta and Dong An},
      year={2025},
      eprint={2508.15170},
      archivePrefix={arXiv},
      primaryClass={quant-ph},
      url={https://arxiv.org/abs/2508.15170}, 
}

@misc{JenningsLostaglioLowrieEtAl2024,
      title={The cost of solving linear differential equations on a quantum computer: fast-forwarding to explicit resource counts}, 
      author={David Jennings and Matteo Lostaglio and Robert B. Lowrie and Sam Pallister and Andrew T. Sornborger},
      year={2024},
      eprint={2309.07881},
      archivePrefix={arXiv},
      primaryClass={quant-ph},
      url={https://arxiv.org/abs/2309.07881}, 
}

@misc{AnChildsLin2023,
      title={Quantum algorithm for linear non-unitary dynamics with near-optimal dependence on all parameters}, 
      author={Dong An and Andrew M. Childs and Lin Lin},
      year={2023},
      eprint={2312.03916},
      archivePrefix={arXiv},
      primaryClass={quant-ph},
      url={https://arxiv.org/abs/2312.03916}, 
}

@misc{AnLiuWangEtAl2022,
  title={A theory of quantum differential equation solvers: limitations and fast-forwarding},
  author={An, Dong and Liu, Jin-Peng and Wang, Daochen and Zhao, Qi},
    archiveprefix = {arXiv},
  eprint        = {2211.05246},
  year={2022}
}

@article{JinLiuYu2024,
   title={Quantum Simulation of Partial Differential Equations via Schrödingerization},
   volume={133},
   ISSN={1079-7114},
   url={http://dx.doi.org/10.1103/PhysRevLett.133.230602},
   DOI={10.1103/physrevlett.133.230602},
   number={23},
   journal={Physical Review Letters},
   publisher={American Physical Society (APS)},
   author={Jin, Shi and Liu, Nana and Yu, Yue},
   year={2024},
   month=Dec }

@article{ChildsWiebe2012,
    author  = {Andrew M. Childs and Nathan Wiebe},
    title   = {Hamiltonian Simulation Using Linear Combinations of Unitary Operations},
    journal = {Quantum Information and Computation},
	doi = {10.26421/qic12.11-12},
	year = 2012,
	publisher = {Rinton Press},
	volume = {12},
    pages   = {901--924},
}

@Article{BrassardHoyerMoscaEtAl2002,
  author  = {Brassard, Gilles and Hoyer, Peter and Mosca, Michele and Tapp, Alain},
  title   = {Quantum amplitude amplification and estimation},
  journal = {Contemp. Math.},
  year    = {2002},
  volume  = {305},
  pages   = {53--74},
  doi = {10.1090/conm/305/05215},
}

@misc{BerryCosta2022,
      title={Quantum algorithm for time-dependent differential equations using Dyson series}, 
      author={Dominic W. Berry and Pedro C. S. Costa},
      year={2022},
      eprint={2212.03544},
      archivePrefix={arXiv},
      primaryClass={quant-ph}
}

@article{AnLiuLin2023,
      doi = {10.1103/physrevlett.131.150603},
  
	year = 2023,
	month = {10},
  
	publisher = {American Physical Society ({APS})},
  
	volume = {131},
  
	number = {15},
  
	author = {Dong An and Jin-Peng Liu and Lin Lin},
  
	title = {Linear Combination of Hamiltonian Simulation for Nonunitary Dynamics with Optimal State Preparation Cost},
  
	journal = {Physical Review Letters}
}

@misc{GroverRudolph2002,
  author = {Grover, Lov and Rudolph, Terry},
  title = {Creating superpositions that correspond to efficiently integrable probability distributions},
  year = {2002},
  journal = {}
}

@article{TakahiraOhashiSogabeEtAl2020,
	doi = {10.26421/qic20.1-2-2},
	year = 2020,
	publisher = {Rinton Press},
	volume = {20},
	number = {1{\&}2},
	pages = {14--36},
	author = {Souichi Takahira and Asuka Ohashi and Tomohiro Sogabe and Tsuyoshi S. Usuda},
	title = {Quantum algorithm for matrix functions by Cauchy{\textquotesingle}s integral formula},
	journal = {Quantum Information and Computation}
}

@article{FangLinTong2022,
   title={Time-marching based quantum solvers for time-dependent linear differential equations},
   volume={7},
   DOI={10.22331/q-2023-03-20-955},
   journal={Quantum},
   publisher={Verein zur Forderung des Open Access Publizierens in den Quantenwissenschaften},
   author={Fang, Di and Lin, Lin and Tong, Yu},
   year={2023},
   month=mar, pages={955} }

@article{GuSommaSahinoglu2021,
	doi = {10.22331/q-2021-11-15-577},
	year = 2021,
	publisher = {Verein zur Forderung des Open Access Publizierens in den Quantenwissenschaften},
	volume = {5},
	pages = {577},
	author = {Shouzhen Gu and Rolando D. Somma and Burak \c{S}ahino{\u{g}}lu},
	title = {Fast-forwarding quantum evolution},
	journal = {Quantum}
}

@article{LiuKoldenKroviEtAl2021,
	doi = {10.1073/pnas.2026805118},
	year = 2021,
	publisher = {Proceedings of the National Academy of Sciences},
	volume = {118},
	number = {35},
	author = {Jin-Peng Liu and Herman {\O}ie Kolden and Hari K. Krovi and Nuno F. Loureiro and Konstantina Trivisa and Andrew M. Childs},
	title = {Efficient quantum algorithm for dissipative nonlinear differential equations},
	journal = {Proceedings of the National Academy of Sciences}
}

@article{Krovi2022,
   title={Improved quantum algorithms for linear and nonlinear differential equations},
   volume={7},
   DOI={10.22331/q-2023-02-02-913},
   journal={Quantum},
   publisher={Verein zur Forderung des Open Access Publizierens in den Quantenwissenschaften},
   author={Krovi, Hari},
   year={2023},
   month=2, pages={913} }

@article{ChildsLiu2020,
	doi = {10.1007/s00220-020-03699-z},
	year = 2020,
	publisher = {Springer Science and Business Media {LLC}
},
	volume = {375},
	number = {2},
	pages = {1427--1457},
	author = {Andrew M. Childs and Jin-Peng Liu},
	title = {Quantum spectral methods for differential equations},
	journal = {Communications in Mathematical Physics}
}

@article{BerryChildsOstranderEtAl2017,
	doi = {10.1007/s00220-017-3002-y},
	year = 2017,
	publisher = {Springer Science and Business Media {LLC}
},
	volume = {356},
	number = {3},
	pages = {1057--1081},
	author = {Dominic W. Berry and Andrew M. Childs and Aaron Ostrander and Guoming Wang},
	title = {Quantum algorithm for linear differential equations with exponentially improved dependence on precision},
	journal = {Communications in Mathematical Physics}
}

@article{Berry2014,
	doi = {10.1088/1751-8113/47/10/105301},
	year = 2014,
	publisher = {{IOP} Publishing},
	volume = {47},
	number = {10},
	pages = {105301},
	author = {Dominic W. Berry},
	title = {High-order quantum algorithm for solving linear differential equations},
	journal = {Journal of Physics A: Mathematical and Theoretical}
}

@Article{BerryChildsCleveEtAl2015,
  Title                    = {Simulating {H}amiltonian dynamics with a truncated {T}aylor series},
  Author                   = {Dominic W. Berry and Andrew M. Childs and Richard Cleve and Robin Kothari and Rolando D. Somma},
  Journal                  = {Phys. Rev. Lett.},
  Year                     = {2015},
  Pages                    = {090502},
  Volume                   = {114}, 
  doi = {10.1103/PhysRevLett.114.090502}
}

@misc{LowWiebe2019,
   title={{H}amiltonian Simulation in the Interaction Picture}, 
      author={Guang Hao Low and Nathan Wiebe},
      year={2019},
      eprint={1805.00675},
      archivePrefix={arXiv},
      primaryClass={quant-ph}
}

@InProceedings{GilyenSuLowEtAl2019,
  author    = {Gily{\'e}n, Andr{\'a}s and Su, Yuan and Low, Guang Hao and Wiebe, Nathan},
  title     = {Quantum singular value transformation and beyond: exponential improvements for quantum matrix arithmetics},
  booktitle = {Proceedings of the 51st Annual ACM SIGACT Symposium on Theory of Computing},
  year      = {2019},
  pages     = {193--204},
  doi = {10.1145/3313276.3316366}
}

@InProceedings{ChakrabortyGilyenJeffery2018,
  author =	{Shantanav Chakraborty and Andr{\'a}s Gily{\'e}n and Stacey Jeffery},
  title =	{{The Power of Block-Encoded Matrix Powers: Improved Regression Techniques via Faster Hamiltonian Simulation}},
  booktitle =	{46th International Colloquium on Automata, Languages, and Programming (ICALP 2019)},
  pages =	{33:1--33:14},
  series =	{Leibniz International Proceedings in Informatics (LIPIcs)},
  ISBN =	{978-3-95977-109-2},
  ISSN =	{1868-8969},
  year =	{2019},
  volume =	{132},
  editor =	{Christel Baier and Ioannis Chatzigiannakis and Paola Flocchini and Stefano Leonardi},
  publisher =	{Schloss Dagstuhl--Leibniz-Zentrum fuer Informatik},
  address =	{Dagstuhl, Germany},
  URL =		{http://drops.dagstuhl.de/opus/volltexte/2019/10609},
  URN =		{urn:nbn:de:0030-drops-106092},
  doi =		{10.4230/LIPIcs.ICALP.2019.33}
}

@article{AtiaAharonov2017,
	author = {Atia, Yosi and Aharonov, Dorit},
	doi = {10.1038/s41467-017-01637-7},
	id = {Atia2017},
	journal = {Nature Communications},
	number = {1},
	pages = {1572},
	title = {Fast-forwarding of Hamiltonians and exponentially precise measurements},
	volume = {8},
	year = {2017}}

\appendix

\section{Bounds of the Carleman error}\label{app:Carleman_error_decay}

Here we show the proofs of the Carleman error, which mainly follow the strategy for the time-independent case established in~\cite{JenningsKorzekwaLostaglioEtAl2025} with suitable modifications for the time-dependent case and more explicit parameter dependence. 

\begin{proof}[Proof of~\cref{lem:matrix_dissipation}]
    
     Let $\nu = [\nu_1;\nu_2;\cdots;\nu_N]$ be an arbitrary normalized vector with the same partition as the Carleman embedding. 
     We can compute
    \begin{align*}
        \nu^\dag (A(t)^\dag + A(t)) \nu 
        & = \sum_{j=1}^{N}  \nu_j^\dag((A_{j}^{j})^\dag + A_{j}^{j}) \nu_j + \sum_{j=2}^{N} \big(\nu_{j-1}^\dag (A_{j-1}^{j})^{\dag} \nu_{j} + \nu_{j}^\dag A_{j-1}^j \nu_{j-1}\big) \nonumber \\
        & \quad + \sum_{j=1}^{N-1} \left(\nu_{j+1}^\dag (A_{j+1}^{j})^{\dag} \nu_{j} + \nu_{j}^\dag A_{j+1}^{j} \nu_{j+1}\right).
    \end{align*}
According to the definitions of $A_{i}^{j}$'s, we can bound each summation as 
\begin{equation*}
    \sum_{j=1}^{N}  \nu_j^\dag ((A_{j}^{j})^{\dag} + A_{j}^{j}) \nu_j(t)  
     \leq \sum_{j=1}^{N} 2j \mu(\Fonegamma)  \| \nu_j\|^2, 
\end{equation*}
\begin{equation*}
    \sum_{j=2}^{N}  \left( \nu_{j-1}^\dag (A_{j-1}^{j})^{\dag} \nu_{j} + \nu_{j}^\dag A_{j-1}^j \nu_{j-1} \right)
    \le \sum_{j=2}^{N} 2 j \| \Fzerogamma \| \| \nu_j\| \| \nu_{j-1}\|, 
\end{equation*}
\begin{equation*}
    \sum_{j=1}^{N-1}  \left(\nu_{j+1}^\dag (A_{j+1}^{j})^{\dag} \nu_{j} + \nu_{j}^\dag A_{j+1}^{j} \nu_{j+1}\right)
    \le \sum_{j=1}^{N-1} 2 j \| \Ftwogamma \| \| \nu_j\| \| \nu_{j+1}\|. 
\end{equation*}
   Define an $N$-dimensional vector $\nu_G$ with $\nu_{G,j} = \|\nu_j\|$ and an $N \times N$ matrix $G$ with non-zero entries $G_{j}^j = 2j\mu(\Fonegamma)$, $G_{j-1}^j = 2j \| \Fzerogamma \|$, $G_{j+1}^j = 2j \| \Ftwogamma \|$, we can rewrite the above bounds as $\nu^\dagger (A^\dagger + A) \nu \leq \nu_G^\dagger G \nu_G$. 
Note that $\nu_G^\dagger G^\dagger \nu_G  =  \nu_G^\dagger G \nu_G$, so we have 
\begin{equation}\label{eqn:proof_decay_A_HG}
    \nu^\dagger(A+A^\dagger)\nu \le \nu_G^\dagger \mathfrak{H}(G) \nu_G \le \lambda_1(\mathfrak{H}(G))\|\nu_G\|^2 = \lambda_1(\mathfrak{H}(G)), 
\end{equation}
where $\mathfrak{H}(G) = \frac{G+G^\dagger}{2}$ and $\lambda_1(\mathfrak{H}(G))$ denotes its largest eigenvalue. 

We then estimate the value of the $\lambda_1(\mathfrak{H}(G))$. 
Recall that all the non-zero entries of $\mathfrak{H}(G)$ are given by 
    \begin{align*}
        \mathfrak{H}(G)_{j-1}^j &= j\| \Fzerogamma \| +(j-1)\| \Ftwogamma \|, \quad 2 \leq j \leq N, \\
        \mathfrak{H}(G)_{j}^j &= 2 j \mu( \Fonegamma ), \quad 1 \leq j \leq N, \\
        \mathfrak{H}(G)_{j+1}^j &= (j+1)\| \Fzerogamma \| + j\| \Ftwogamma \|, \quad 1 \leq j \leq N-1. 
    \end{align*}
    By applying the Gershgorin's circle theorem, we have 
    \begin{equation*}
        \lambda_1(\mathfrak{H}(G))\le \max_{1 \leq j \leq N}(c_j+r_j), 
    \end{equation*}
    where
    \begin{align*}
        & c_j = 2j\mu( \Fonegamma ),\quad 1 \leq j \leq N, \\
        & r_j = (2j+1)\|\Fzerogamma\|+(2j-1)\|\Ftwogamma\|, \quad 2 \leq j \leq N-1,  \\ 
        & r_1 = 2\|\Fzerogamma\| + \|\Ftwogamma\|,\quad 
          r_N = N\|\Fzerogamma\|+(N-1)\|\Ftwogamma\|. 
    \end{align*}
    For convenience, we denote $a = \|{F}_2\|$, $b = -\mu({F}_1)$, $c = \|{F}_0\|$, $r = a+c-b$, and $x = \|u(0)\|$. 
    Based on the rescaling factor $\gamma$ introduced in the previous linearization procedure
    \begin{equation}\label{eqn:proof_decay_gamma_choice}
        \frac{a}{\gamma} + c\gamma =\frac{1+R}{2}b,\quad  \mu(\Fonegamma) = -b ,\quad \|\Fzerogamma\| = \gamma c, \quad \|\Ftwogamma \| = \frac{a}{\gamma}. 
    \end{equation}
    Then, since $R < 1$, we have 
    \begin{equation*}
        \overline{r} \coloneqq \frac{a}{\gamma} + c\gamma - b = \frac{R-1}{2}b < 0,  
    \end{equation*}
    and thus 
    \begin{align}
        \lambda_1(\mathfrak{H}(G)) &\le \max_{1 \leq j \leq N}(c_j+r_j) = \max\left\{ c_1+r_1, \max_{2 \leq j \leq N}(c_j+r_j) \right\} \nonumber \\
        & \le \max\left\{ 2\overline{r} - a/\gamma, \max_{2 \leq j \leq N}( 2j\overline{r} + c\gamma -a/\gamma ) \right\} \nonumber \\
        & = \max\left\{ 2\overline{r} - a/\gamma, 4\overline{r} + c\gamma -a/\gamma \right\}. \label{eqn:proof_decay_eig_HG}
    \end{align}
    By the definition of $R$ and the choice of $\gamma$ to be the smaller root of~\cref{eqn:proof_decay_gamma_choice}, we have 
    \begin{equation*}
        2\overline{r} - \frac{a}{\gamma} = \frac{a}{\gamma} + 2c\gamma - 2b \leq 2 \left( \frac{a}{\gamma} + c\gamma - b \right) = -(1-R)b, 
    \end{equation*}
    and
    \begin{equation*}
        4\overline{r} + c\gamma -a/\gamma = 3 \frac{a}{\gamma} + 5c\gamma - 4b \leq 4\left( \frac{a}{\gamma} + c\gamma - b \right) = -2(1-R)b, 
    \end{equation*}
    and~\cref{eqn:proof_decay_eig_HG} becomes 
    \begin{equation*}
        \lambda_1(\mathfrak{H}(G)) \leq -(1-R)b. 
    \end{equation*}
    Plugging this back to~\cref{eqn:proof_decay_A_HG} yields the desired bound for $A(t)$. 
\end{proof}

\begin{proof}[Proof of~\cref{lem:solution_norm_bound}]

    Following the notations in the proof of~\cref{lem:matrix_dissipation}, let 
        $a = \|{F}_2\|$, $b = -\mu({F}_1)$, $c = \|{F}_0\|$, $r = a+c-b$, $x = \|u(0)\|$. 
    According to the definition of $R$ and the choice of the rescaling factor in~\cref{eqn:choice_gamma}, we can write 
    \begin{equation*}
        \hat{\Gamma} +\frac{1}{\hat{\Gamma}} = B\frac{1+R}{2} \coloneqq B_1, \quad \hat{x} +\frac{1}{\hat{x}} = BR \coloneqq B_2, 
    \end{equation*}
    where 
    \begin{equation*}
        \hat{x} = \sqrt{\frac{a}{c}} \,x,\quad \hat{\Gamma} = \sqrt{\frac{a}{c}}\,\frac{1}{\gamma}, \quad B = \frac{b}{\sqrt{ac}}. 
    \end{equation*}
    Then we have 
    \begin{equation*}
        \|u_{\gamma}(0)\| = \gamma x = \frac{\hat{x}}{\hat{\Gamma}} \leq \frac{B_2 + \sqrt{ B_2^2 - 4 }}{B_1 + \sqrt{B_1^2 - 4}} < \frac{B_2}{B_1} = \frac{2R}{R+1}. 
    \end{equation*}

    To bound the solution norm at any time $t$, we first compute that 
    \begin{equation*}
        \frac{d\|u(t)\|}{dt} = \frac{1}{2\|u(t)\|}\frac{d\|u(t)\|^2}{dt} = \frac{1}{\|u(t)\|}\mathcal{R}e \langle \frac{du(t)}{dt}, u(t) \rangle
    \end{equation*}
    where $\langle \cdot, \cdot \rangle$ denotes the inner product of two (possibly unnormalized) vectors. 
    Then we have
    \begin{equation*}
        \frac{d\|u(t)\|}{dt} = \frac{1}{\|u(t)\|} \mathcal{R}e \langle F_2(t) u(t)^{\otimes 2} + F_1(t) u(t) + F_0(t), u(t) \rangle
    \end{equation*}
    Notice that 
    \begin{align*}
        \mathcal{R}e \langle F_2(t) u(t)^{\otimes 2} ,u(t) \rangle &\leq \|F_2(t) u(t)^{\otimes 2}\|\|u(t)\| \leq a\|u(t)\|^3, \\
        \mathcal{R}e \langle F_0(t) ,u(t) \rangle &\leq \|F_0(t)\|\|u(t)\| \leq c\|u(t)\|, \\
        \mathcal{R}e \langle F_1(t) u(t), u(t) \rangle &\leq -b \|u(t)\|^2, 
    \end{align*}
    so we have 
    \begin{equation}\label{eq:proof_solu_norm_du}
        \frac{d\|u(t)\|}{dt} \leq a\|u(t)\|^2 - b \|u(t)\| + c. 
    \end{equation}
    Consider the ODE
    \begin{equation}\label{eq:proof_solu_norm_dv}
        \frac{dv(t)}{dt} = a v(t)^2 - b v(t) + c, \quad v(0) = \|u(0)\|. 
    \end{equation}
    Since $R < 1$, which is equivalent to $a\|u(0)\|^2 - b \|u(0)\| + c < 0$, the quadratic equation $a y^2 - b y + c = 0$ must have two different roots denoted as $0<v_- <\|u(0)\|< v_+$. 
    Let $w(t) = v(t) - v_-$, then we have 
    \begin{equation*}
        \frac{dw(t)}{dt} = a w(t)(w(t) -(v_+ - v_-)), \quad 0< v(0) - v_- = w(0)  < v_+-v_-. 
    \end{equation*}
    We can analytically solve the above ODE and get
    
    \begin{equation*}
    w(t)
    =
    \frac{w(0)(v_+-v_-)}
    {e^{a(v_+-v_-)t}(v_+-v_-)
    -w(0)\bigl(e^{a(v_+-v_-)t}-1\bigr)}.
    \end{equation*}
    Since $0<w(0)<v_+-v_-\,,a>0\,,t\geq 0$ ,we have
    \begin{align*}
    &e^{a(v_+-v_-)t}(v_+-v_-)
    -w(0)\bigl(e^{a(v_+-v_-)t}-1\bigr)\\
    &
    =\bigl(e^{a(v_+-v_-)t}-1\bigr)\bigl(v_+-v_--w(0)\bigr) + v_+-v_-
    \geq v_+-v_-.
    \end{align*}
    Therefore, for $t \geq 0$, we have $0<w(t)\leq \frac{w(0)(v_+-v_-)}{v_+-v_-} =w(0)$, and thus $v(t) = w(t) + v_- \leq w(0) + v_- = v(0) = \|u(0)\|$. 
    According to the comparison theorem between ODEs~\cref{eq:proof_solu_norm_du} and~\cref{eq:proof_solu_norm_dv}, we have $\|u(t)\| \leq v(t) \leq \|u(0)\|$, and thus $\|\ugamma(t)\| \leq \|\ugamma(0)\|$. 
\end{proof}

\begin{proof}[Proof of~\cref{lem:Carleman_error_bound}]
    By Duhamel's principle, we write the error vector $\eta$ from~\cref{eq:error_evolution} as 
    \begin{equation*}
        \eta(t) = \int_0^t  \mathcal{T} e^{\int_s^t A(\tau) d\tau} \zeta(s)  ds. 
    \end{equation*}
    According to~\cref{lem:matrix_dissipation} and~\cref{lem:solution_norm_bound}, we can bound $\eta$ as 
    \begin{align*}
        \| \eta(t) \| &\leq \int_0^t \left\|\mathcal{T} e^{\int_s^t A(\tau) d\tau}\right\| \|\zeta(s)\| ds \leq \int_0^t e^{ \frac{(1-R)\mu(F_1)}{2} (t-s) } N \frac{\|F_2\|}{\gamma} \|\ugamma (s)\|^{N+1} ds \\
        & \leq \frac{2 N \|F_2\| }{-(1-R)\gamma \mu(F_1)} \left(\frac{2R}{1+R}\right)^{N+1} \leq \frac{ 2 N R }{1-R} \left(\frac{2R}{1+R}\right)^{N}, 
    \end{align*}
    where in the last inequality we use an estimate from~\cref{eqn:choice_gamma} that $-\gamma \mu(F_1) = 2(\|F_2\| + \|F_0\|\gamma^2)/(1+R) \geq 2\|F_2\|/(1+R)$. 
\end{proof}

\end{document}